\documentclass[12pt]{report}
\title{\textbf{Refactoring Software Packages via Community Detection from Stability Point of View}}
\author{Mohammad Reza Ahmadzadeh Raji}
\date{\bigskip\bigskip
	A Thesis Presented in Partial Fulfilment of the Requirements for the Degree of Master of Engineering in Computer Engineering
	\\ Razi University \\
	\bigskip
	2014
	}

\usepackage{graphicx}
\usepackage{cite}
\usepackage{amsthm}
\usepackage{amsmath}
\usepackage{url}
\usepackage{amssymb}
\usepackage{fncychap}
\usepackage{algorithm}
\usepackage[noend]{algpseudocode}
\usepackage[titletoc]{appendix}
\usepackage{color}
\usepackage{xcolor}
\usepackage{listings}
\usepackage{textcomp}
\usepackage{mathtools}
\usepackage{setspace}
\usepackage[top=25mm, bottom=25mm, left=30mm, right=25mm]{geometry}
\usepackage{geometry}
\usepackage{subfig}

\definecolor{mygreen}{rgb}{0,0.6,0}
\definecolor{mygray}{rgb}{0.5,0.5,0.5}
\definecolor{mymauve}{rgb}{0.58,0,0.82}

\ChNameVar{\thispagestyle{empty}\centering\Huge\rm\bfseries}

\lstdefinelanguage{JavaScript}{
  keywords={typeof, new, true, false, catch, function, return, null, catch, switch, var, if, in, while, do, else, case, break},
  keywordstyle=\color{blue}\bfseries,
  ndkeywords={class, export, boolean, throw, implements, import, this},
  ndkeywordstyle=\color{darkgray}\bfseries,
  identifierstyle=\color{black},
  sensitive=false,
  comment=[l]{//},
  morecomment=[s]{/*}{*/},
  commentstyle=\color{purple}\ttfamily,
  stringstyle=\color{red}\ttfamily,
  morestring=[b]',
  morestring=[b]"
}

\newtheorem{defi}{Definition}
\newtheorem{theorem}{Theorem}[section]

\newtheorem{remark}[theorem]{Remark}
\newtheorem{prop}[theorem]{Proposition}


\makeatletter
\def\BState{\State\hskip-\ALG@thistlm}
\makeatother

\begin{document}

\thispagestyle{empty}
\begin{center}

\includegraphics[scale=.4]{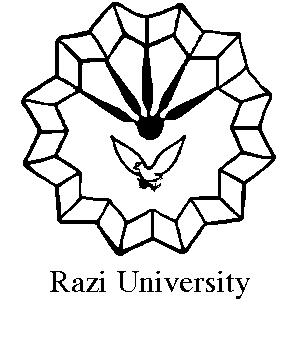} \\

\fontsize{11pt}{12pt}\textbf{Faculty of Engineering} \\ 
\fontsize{10pt}{12pt}\textbf{Department of Computer Engineering}\\[2cm]

\fontsize{13pt}{12pt}\textbf{M.Sc.Thesis} \\ [6mm]
\vskip 2.5cm

\fontsize{15pt}{12pt}
{
	\textbf{{Refactoring Software Packages via Community Detection \vspace{0.3cm}\\ from Stability Point of View}}
}

\vskip 1.5cm
\fontsize{13pt}{12pt}
{
	\textbf{Supervisor:} \\ \textbf{Dr. Behzad Montazeri}
	\vskip 1.5cm
	\vskip 1.5cm
	\textbf{By}         \\ \textbf{Mohammad Reza Ahmadzadeh Raji}
}
\vskip 1.5cm
\textbf{June 2014}

\end{center}

\pagebreak

%
%
%
%


\pagebreak
\chapter*{Acknowledgements}
\thispagestyle{empty}
\paragraph{}
First and foremost, I would like to express my appreciation and thanks to my supervisor, Dr. Behzad Montazeri. I would like to thank him for his encouragement, motivation and priceless advices throughout my research. 
\paragraph{}
I would also like to express my deepest appreciation to my father, Dr. Mehrdad Ahmadzadeh, my mother and my wife, Sara who have been the most supportive, through thick and thin and have helped me in the difficulties of the paths I have taken.

\pagebreak
\thispagestyle{empty}
\begin{center}
	\vspace*{\fill}
		\large\textbf{Dedicated to my beloved wife, for her endless support and encouragement}
	\vspace*{\fill}
\end{center}

\pagebreak
\chapter*{Abstract}
\thispagestyle{empty}
\paragraph{}
As the complexity and size of software projects increases in real-world environments, maintaining and creating maintainable and dependable code becomes harder and more costly. Refactoring is considered as a method for enhancing the internal structure of code for improving many software properties such as maintainability. 
\paragraph{}
In this thesis, the subject of refactoring software packages using community detection algorithms is discussed, with a focus on the notion of package stability. The proposed algorithm starts by extracting a package dependency network from Java byte code and a community detection algorithm is used to find possible changes in package structures. In this work, the reasons for the importance of considering dependency directions while modeling package dependencies with graphs are also discussed, and a proof for the relationship between package stability and the modularity of package dependency graphs is presented that shows how modularity is in favor of package stability.
\paragraph{}
For evaluating the proposed algorithm, a tool for live analysis of software packages is implemented, and two software systems are tested. Results show that modeling package dependencies with directed graphs and applying the presented refactoring method, leads to a higher increase in package stability than undirected graph modeling approaches that have been studied in the literature. 
\\\\
\textbf{Keywords: }Graph clustering, Community detection, Package refactoring, Software metrics, Stability, Coupling, Cohesion.

\pagenumbering{Alph}
\tableofcontents

\listoffigures
\listoftables

\cleardoublepage\pagenumbering{arabic}
\chapter{Introduction to code refactoring}
\pagebreak
\paragraph{}
There are many properties that can be associated with good code. Sommerville describes good code as one that is highly maintainable, dependable, efficient and usable \cite{sommerville2004}. Truly reusable code is considered gold in the software industry as it significantly effects productivity and thus lowers costs \cite{lim1994} and without a doubt, good code is backed by a good design. A professional software engineer must first design a software and then implement the code based on the design. However, in real-world scenarios, the great attributes of a good software might fade away as the project grows. Tight schedules, high customer demands and the high number of programmers involved in large projects are considered as some of the reasons that make efficient and engineered implementations change into a mess. A mess that is not easily maintained, reused, changed or depended upon. Refactoring is considered the cure for this infiltration of the project. 

\paragraph{}
Refactoring is a common word for a day-to-day programmer with its origins in mathematics and ultimately in the Latin language. The root \textit{factor} has the meaning of \textit{maker} and hence refactoring is known as re-making something. In mathematics, when you factor an expression, you re-make it and provide a more cleaner version. The exact origin of the word, refactoring, in computer science is somewhat unknown, however the Forth language community is known to have been the first people to have used this expression \cite{fowler1999}. Chapter six of Leo Brodie's book, Thinking Forth is dedicated to the subject of refactoring \cite{brodie1984}. 
\paragraph{}
Martin Fowler, the author of one of the most canonical books on refactoring \cite{fowler1999}, describes it as \textit{``the process of changing a software system in such a way that it does not alter the external behavior of the code, yet improves its internal structure.''}
\paragraph{}
This thesis solely focuses on refactoring methods that involve the use of graph clustering methods, however to better understand the reasons and effects of proposed methods, a brief and concise explanation of known refactoring techniques is given.

\section{Well-known refactoring techniques}

\begin{itemize}
\item \textbf{Rename method.} This technique may be the most simple refactoring method one can use. Simply renaming identifiers and variables will make the code clearer, more understandable and can reduce the need for comments. An appropriate name for a method, variable or a class is one that is descriptive so that a new programmer can understand its work just by a glance.

\item \textbf{Inline temp.} Temporary variables can make methods longer and more complicated. It is suggested that temporary variables that are being used only once or are a result of a method call be completely removed and the value assigned to them be used in the code. An example is provided below.

\textbf{Incorrect:}
\begin{lstlisting}[language=python]
	def add_something():
		return 1 + 2
	
	def foo():
		temp_variable = add_something()
		print "The result is " + temp_variable
\end{lstlisting} 
\textbf{Correct:}
\begin{lstlisting}[language=python]
	def foo():
		print "The result is " + add_something()
\end{lstlisting} 

\item \textbf{Extract method. }Known as arguably the most important refactoring technique, Extract method aims at reducing the size of long methods by breaking them into smaller methods with descriptive variables. Many refactoring and simplifying techniques in software engineering involve breaking code and algorithms into smaller, more understandable chunks. This method is one of them. 

\begin{lstlisting}[language=python]
	class Foo:
		username
		def __init__(self):
			# Some initialization code
			self.username = "Some username"
	
		def func1(self):
			print "Welcome"
			print "You have logged in as " + self.username
			print "Something else"
			
		def func2(self):
			print "Welcome"
			print "You have logged in as " + self.username
			print "Some reports"
			
\end{lstlisting} 
\paragraph{}
In the provided example, lines 8 and 9 are equal to lines 13 and 14 and can be extracted into a new method that greets the user.
Extract method is considered as an important and basic refactoring technique that highly effects the cohesion of classes from which methods have been extracted. Extract method suggests the extraction of pieces of code that are used more than once (duplicate code). If this condition is met while extracting piece of code A from methods B and C, then after refactoring, both B and C will be using A and thus reducing the cohesion in their class. However, one must realize that if appropriate interfaces are not used in the code and other classes in a package use method A, then instead of reducing cohesion, coupling will be increased. A thorough study on this issue and a metric for finding appropriate pieces of code for extracting while considering the notion of cohesion is provided in \cite{goto2013extract}.

Considering our focus on graph clustering methods in refactoring, it is worth noting that some work has been done in detemining the class a method belongs to, with the help of community detection techniques \cite{pan2009class}. However, introducing new methods and extracting them with community detection is still in need of attention. 

\item \textbf{Inline method.} In some cases, the opposite of Extract method should be applied. Suppose method A is simple, clear and is being used only once, possibly in a stable class whose content is not likely to change. In this case, using an identifier for the code in method A only results in an extra call for no benefit. This method can be removed and its content can be used inline. 

\item \textbf{Replace method with Method object. } This technique can be considered as an aid, in situations where Extract method becomes difficult because of the high number of temporary variables in a long method. In a case where the number of temporary variables is high, Extract method can become a cumbersome task because passing around all the temporary variables between the extracted methods can prove to be messy and finding the needed temporary variables for a piece of extracted code can take a lot of time.

To resolve this issue, one approach is to move the long method into a new class, set the local temporary variables as class attributes and then apply Extract method. This method provides a better state, from which we can continue our refactoring using Extract method or other techniques. 

\item \textbf{Pull up method. }Imagine a scenario in which a piece of code is duplicated in two different classes, it is best to pull that code up into a super class of those two classes. 

\textbf{Before refactoring:}
\begin{lstlisting}[language=python]
	class Person:
		firstname = None
		lastname = None
		def __init__(self):
			# Some initialization code
			
	class Student(Person):
		studentNo = None
		def __init__(self):
			# Some initialization code
			
		def makeFullName(self):
			return self.firstname + " " + self.lastname
			
		def getStudentNo(self):
			return self.studentNo
			
	class Employee(Person):
		salary = None
		def __init__(self):
			# Some initialization code
		
		def makeFullName(self):
			return self.firstname + " " + self.lastname
			
		def getSalary(self):
			return self.salary
			
\end{lstlisting} 

\textbf{After refactoring:}
\begin{lstlisting}[language=python]
	class Person:
		firstname = None
		lastname = None
		def __init__(self):
			# Some initialization code
	
		def makeFullName(self):
			return self.firstname + " " + self.lastname
			
	class Student(Person):
		studentNo = None
		def __init__(self):
			# Some initialization code
			
		def getStudentNo(self):
			return self.studentNo
			
	class Employee(Person):
		salary = None
		def __init__(self):
			# Some initialization code
			
		def getSalary(self):
			return self.salary
			
\end{lstlisting} 

\item \textbf{Extract surrounding method. } Imagine a case in which several different methods are almost identical but with a slight difference in the middle of each one. In some languages, one can pull up the duplicated code into a new method and pass the middle section to the method which it yields to. This ability is provided in some languages like Ruby and can be simulated in some other languages by passing callback functions. A Ruby example is given below.

\begin{lstlisting}[language=ruby]
	def testMethod
		puts "Something printed from inside testMethod"
		yield
		puts "Something printed from inside testMethod"
	end
	testMethod {puts "Something printed from the block"}
\end{lstlisting} 

\item \textbf{Replace conditional with polymorphism.} This method of refactoring can help remove the complexity and code smell of conditional logic and demonstrate the principle of true object-oriented design. 

\end{itemize}

\chapter{Code quality and software metrics}
\thispagestyle{empty}	
\pagebreak
\paragraph{}
Code quality is one of the most important factors that directly effects a project's maintainability phase. The quality of a good code determines how flexible a project is and determines the limits of the resuability of its components. Good code can be changed more quickly and newcomers to the project can understand it more easily. Another important feature that code quality can greatly effect, is how much we can trust that new changes will not cause unexpected effects and will not introduce bugs to the project \cite{baggen2012standardized}. 

\paragraph{}
Although code quality is considered as an important subject in software projects, many developers and projects simply neglect its importance and features. This is normally due to tightened schedules, high customer demands and low budgets. 

\paragraph{}
Code quality has been discussed in the literature for a long time and more specifically, quantitative measurement of quality factors and providing metrics has been greatly studied. However, the study of software metrics has been diverse and the need for a strong and refined approach is felt \cite{kitchenham2010s}. In this thesis we focus on three well-known object-oriented quality factors, namely coupling, cohesion and package stability. 

\section{Coupling and cohesion}
\paragraph{}
Coupling is one of the most famous internal product attributes. Generally, two pieces of code are said to be coupled if changes in one causes the other to change. In the object-oriented paradigm, coupling between two classes is considered a bad and unwanted attribute, however a system with no coupling between its classes would mean that interaction is not occurring between the classes and therefore it would simply fail to function. 
\paragraph{}
Cohesion, which almost always comes with coupling, is another important internal product attribute. In an object-oriented system, a class is said to have a high cohesion if its internal structures and methods have high connectivity with themselves. The goal in a good design is high cohesion and low coupling, meaning that classes should be cohesive and therefore fully related to their responsibility while they have a low coupling with other classes so that they can change without causing too many changes in other parts of the system. Designs with high cohesion and low coupling make the system more reliable and maintainable \cite{fenton1998software}, \cite{troy1981measuring}.

\paragraph{}
The notions of coupling and cohesion have been excessively studied in the literature and many metrics have been proposed for measuring them. This thesis briefly surveys different approaches in the literature. 

\subsection{Basic definitions by Myers}
\paragraph{}
Myers, Stevens and Constantine introduced the concept of coupling in procedural programming. Based on this, Fenton defined six different levels of coupling \cite{alghamdi2008measuring}. These levels of coupling are shown below from worst to best. 

\begin{itemize}
\item \textbf{Content coupling. }If one element branches into or changes the internal statements of another element, they are said to have content coupling. 
\item \textbf{Common coupling. }If two elements refer to the same global variable, they are said to have common coupling. 
\item \textbf{Control coupling. }If the data that one element sends to the other controls its behavior, then control coupling is implied. 
\item \textbf{Stamp coupling. }Two elements are stamp coupled if they send more information to each other than necessary. 
\item \textbf{Data coupling. }If two elements communicate with each other by parameters with no control coupling, then they are data coupled. 
\item \textbf{No coupling. }If two elements have no communication with each other then they are not coupled. 
\end{itemize}

\subsection{Fenton and melton's metric}
\paragraph{}
Fenton and melton proposed a metric for coupling which is expressed as

\begin{equation}
\label{eq:fenton}
C(x, y) = i + n / (n+1)
\end{equation}

where $n$ is the number of interactions between the two components $x$ and $y$ and $i$ is the level of the worst coupling found between $x$ and $y$. In their metric, the coupling level is based on Myer's classification. No coupling is given a coupling level of 0 and the next levels have higher numeric values.

\paragraph{}
Alghamdi discusses several important points about this metric \cite{alghamdi2008measuring}.

\begin{itemize}
\item All types of interconnections are considered equal, with equal effects on coupling. 
\item The Fenton and Melton metric is an example of a inter-modular metric, meaning that it calculates the coupling between a pair of components in contrast with intrinsic metrics that measure the coupling of a component individually. 
\item Coupling values approach the next level when the interconnections between two components increases. 
\end{itemize}

\paragraph{}
Alghamdi also proposes a new coupling metric based on a description matrix of the system \cite{alghamdi2008measuring}.

\subsection{Chidamber and Kemerer's suite}
\paragraph{}
Chidamber and Kemerer \cite{chidamber1991towards} gives the first formal definition of coupling by defining coupling as any evidence that a method from one class uses a method or variable of another class. In their proposed suite, known as the CK suite, Chidamber and Kemerer give provide different metrics. The six metrics are as follows.

\begin{itemize}
\item Weighted Method per Class (WMC)
\item Number Of Children (NOC)
\item Depth of Inheritence Tree (DIT)
\item Coupling Between Objects (CBO)
\item Lack of Cohesion in Methods (LCOM)
\item Response For a Class (RFC)
\end{itemize}

\paragraph{}
Among their six metrics, CBO (Coupling Between Objects) is proportional to the number of non-inheritance related couples with other classes. For measuring coupling, CBO aggregates the total number of couples a class has to another class, which implies that different couples have the same strength and effect. Hitz and Montazeri \cite{hitz1995measuring} argue that the CK suite does not fully conform to measurment theory. 

\subsection{Alghamdi's coupling metric}
\paragraph{}
Alghamdi's approach is based on the idea of generating a description matrix of all the factors that effect coupling and then calculating a coupling matrix based on the collected data. An overview of this approach is depicted in Fig. \ref{alghamdi}

\begin{figure}[ht!]
\centering
\includegraphics[width=.7\columnwidth]{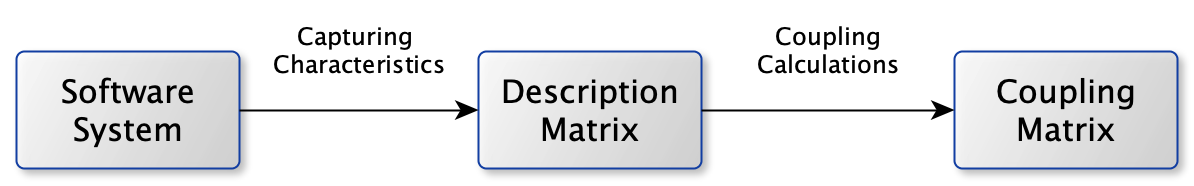}
\caption{An overview of Alghamdi's coupling metric}
\label{alghamdi}
\end{figure}

\paragraph{}
The description matrix is an $m$ by $n$ matrix where $m$ is the number of system components and $n$ is the number of component members. In an object-oriented system, components are represented by classes and members are class variables and methods. An example of a description matrix is depicted in Table \ref{desc_matrix}.

\begin{table}
	\caption{An example of Alghamdi's description matrix}
	\label{desc_matrix}
	\centering
    \begin{tabular}{|l|l|l|l|l|}
    \hline
    \textbf{Component} & $E_1$  & $E_2$ & ... & $E_n$ \\ \hline
    $C_1$        & $d_{11}$ & $d_{12}$  & ~   & $d_{1n}$  \\ \hline
    $C_2$        & $d_{21}$ & $d_{22}$  & ~   & $d_{2n}$  \\ \hline
    ...       	& ~   & ~  & ~   & ~  \\ \hline
    $C_m$        & $d_{m1}$ & $d_{m2}$  & ~   & $d_{mn}$  \\ \hline
    \end{tabular}
\end{table}

\subsection{A qualitative approach to coupling and cohesion}
\paragraph{}
While many quantitative approaches for measuring coupling and cohesion have been proposed in the literature, few qualitative approaches have been discussed. Kelsen \cite{representational2003} proposes an interesting information-based method for analyzing coupling and cohesion and finding refactoring suggestions. Kelsen's approach considers a special type of coupling, namely representational coupling. When an object calls a method of another object, some information about the callee is exposed. If the information is about the low level implementation of the callee, then representational coupling is high. If the call exposes higher level information, then representational coupling is low. Many metrics in the literature, including the works of Chidamber and Kemerer \cite{chidamber1991towards} can not capture representational coupling \cite{representational2003}. The main reason behind this issue is that many works simply count different types of interactions and assign ordinal numbers to these interactions. Kelsen, also presents a minimum for the representational coupling inherently contained in a system, which is known as intrinsic representational coupling. 

\paragraph{}
Kelsen's approach is based on the idea that if one can find two states in a system, namely \textit{witness states}, that yield different messages between objects but cannot affect the states of other objects, then this indicates that coupling can be improved and representational coupling is higher than it should. The elevator example below is borrowed from \cite{representational2003}.

\paragraph{}
Suppose that the behavior of some elevators in a building is modeled using two classes, ElevatorControl and Elevator. Every elevator has two methods, \textit{direction()} and \textit{position()}, which return the direction and position of the elevator. The ElevatorControl class is responsible for handling requests	and checks every elevator's distance and position for finding the closest elevator for a request. Two different implementations for ElevatorControl's \textit{handleRequest} method can be written. 

\begin{algorithm}[h!]
\caption{Implementation 1 for ElevatorControl.handleRequest}\label{alg:qualitative_imp1}
\begin{algorithmic}[1]	
\Procedure{handleRequest}{Request r}	
\State $minDist \gets Infinity$
\State Elevator $ec \gets null$ //reference to closes elevator
\For {every elevator $e$}
	\State $d \gets $ Compute $distance(e, r)$ using $e.direction()$ and $e.position()$
	\If {$d < minDist$}
		\State $ec \gets e$
		\State $minDist \gets d$
	\EndIf
\EndFor
\If {$ec \text{ is not } null$}
	\State $ec.addRequest(r)$
\EndIf
\EndProcedure
\end{algorithmic}
\end{algorithm}

\begin{algorithm}[h!]
\caption{Implementation 2 for ElevatorControl.handleRequest}\label{alg:qualitative_imp2}
\begin{algorithmic}[1]	
\Procedure{handleRequest}{Request r}	
\State $minDist \gets Infinity$
\State Elevator $ec \gets null$ //reference to closes elevator
\For {every elevator $e$}
	\State $d \gets e.computeDistance(r)$
	\If {$d < minDist$}
		\State $ec \gets e$
		\State $minDist \gets d$
	\EndIf
\EndFor
\If {$ec \text{ is not } null$}
	\State $ec.addRequest(r)$
\EndIf
\EndProcedure
\end{algorithmic}
\end{algorithm}

\paragraph{}
In the second implementation, the task of computing en elevator's distance is given to each elevator with the $computeDistance$ method. It is clear that the $ElevatorControl$ class does not need to know the distance and position of every elevator and only needs their distances, therefore less information from the $Elevator$ class is being exposed in the second implementation and representational coupling is decreased. 

\paragraph{}
Kelsen's qualitative approach may be considered a precise and great method for measuring representational coupling, however, because of its non-quantitative nature it is not clear if it can be applied to large, real-life software systems with many classes, and its utilization in real-life scenarios is currently considered as an open problem. 

\section{Stability}
\paragraph{}
Stability is the amount of likeliness, that a class or a package will not change. Stability is inherently difficult to measure because the future changes and needs of a project are not well known, however some metrics exist that try to measure stability. The importance of this stability in software metrics was first mentioned by Hitz and Montazeri \cite{hitz1995measuring}. 

\paragraph{}
Some methods for measuring the stability of a software package, utilize the history of the class's changes in the past and try to predict its future. The changes of a class or a package is typically accessed through version control systems such as Git\footnote{\url{http://git-scm.com}} and Subversion\footnote{\url{http://subversion.apache.org}}, however these approaches can not be used in early stages of software design because of the lack of change history available at the time.

\paragraph{}
Robert Martin \cite{martin2003agile} takes a different approach to measuring the stability of a software package. He believes that stability is proportional to responsibility and a package is said to be responsible and independent if many other entities depend on it, while it doesn't depend on others itself. A package $p$ is said to be irresponsible and thus unstable, if it depends on many other entities, meaning that if they change they cause $p$ to change as well. By Martin's definition, in Fig. \ref{stable_package}, X is an example of a stable package and in Fig. \ref{unstable_package}, Y resembles an unstable package. 

\begin{figure}[ht!]
\centering
\includegraphics[scale=0.3]{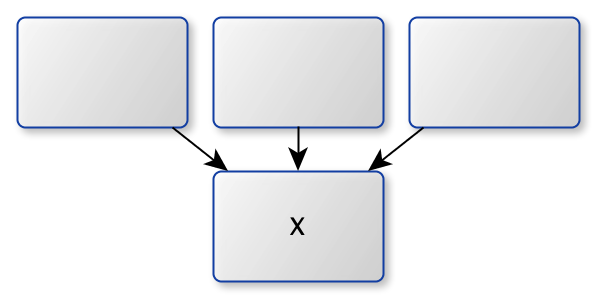}
\caption{An example of a stable package}
\label{stable_package}
\end{figure}

\begin{figure}[ht!]
\centering
\includegraphics[scale=0.3]{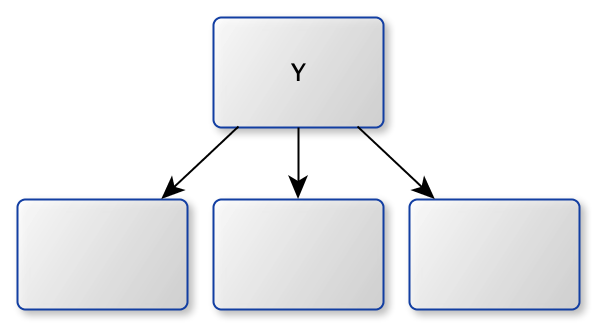}
\caption{An example of an unstable package}
\label{unstable_package}
\end{figure}

\paragraph{}
As a metric for stability, Martin defines the instability of a package as given in Eq. \ref{eq:Instability} where $I$ is instability, $C_a$ is afferent couplings and $C_e$ represents the number of efferent couplings. Afferent couplings is the number of classes outside the package that depend on classes within the package and efferent couplings is the number of classes within the package that depend on outside classes. 

\begin{equation}
\label{eq:Instability}
I=\frac{C_e}{C_a + C_e}
\end{equation}

\paragraph{}
If a package $p$ has an \textit{instability} of 0, then the $p$ has maximum stability and if the package holds a value of 1 for \textit{instability}, then it would mean that the number of afferent couplings is 0 and therefore $p$ depends on other packages while no other package depends on $p$ and this would make it an extremely unstable package.

\paragraph{}
Martin also proposes the Stable Dependencies Principle (SDP) that helps the software design process by ensuring that modules that should be easily changeable not depend on modules that are harder to change \cite{martin2003agile}. In this case, packages should always have a higher $I$ metric than the ones they depend on. Concenting to this principle, one would be able to see a tree of packages, in which stable ones are placed at the bottom and the most unstable ones are at the top. The benefit of this approach is that packages that are violating SDP can be easily spotted. Any package depending on a package above it, would mean a violation of the principle. 

\paragraph{}
It is important to note that not all packages should or could be fully stable, as this would cause an unchangeable and inflexible system. Also, not all packages can be unstable as this would create an irresponsible system with a large number of connections and a high coupling. It is clear that pieces of code that are likely to change should be placed into unstable packages and pieces of code that are not very likely to change in the future should be placed in stable packages. Martin argues that high level design can not be placed in unstable packages because it resembles the architectural decisions of the projects, however if high level code is placed in stable packages then it would almost be impossible to change it after the project becomes more mature and more pieces of code start depending on it. The solution to this dilemma is the use of abstract classes that can introduce the flexibility and flow of stability that is needed. The basic idea behind the Stable Abstraction Principle (SAP) is that a package has to be as abstract as it is stable. This principle ensures that  the stability of a package does not contradict its flexibility. The SAP proposes a metric for measuring the abstractness of a package which is a simple ratio and is shown in Eq. \ref{eq:abstractness} in which $N_a$ is the number of abstract classes inside the package and $N_c$ is the number of classes inside the package. 

\begin{figure}[ht!]
\centering
\includegraphics[scale=0.4]{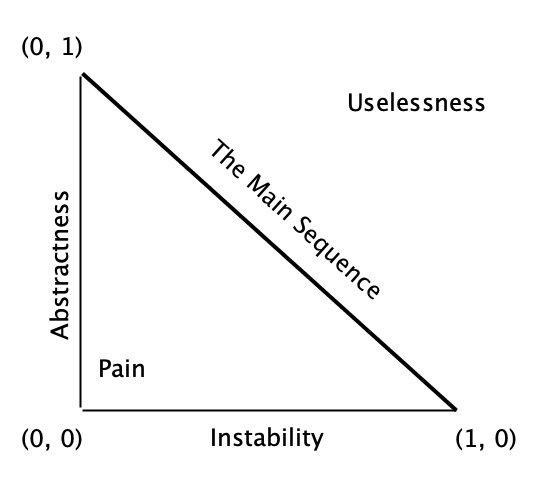}
\caption{The relationship between asbtractness and instability}
\label{abstract_stable_graph}
\end{figure}

\paragraph{}
Martin defines three important areas in the relationship between abstractness and stability. If we set abstractness (A) as the vertical axis and instability (I) as the horizontal axis in a cartesian graph, then three spots depicted in Fig. \ref{abstract_stable_graph} are as follows.

\begin{itemize}
\item \textbf{Zone of pain.} The zone of pain is where a package is highly stable and yet its abstractness is zero. Such a package is hardly changeable.
\item \textbf{Zone of uselessness.} A package in this zone is highly abstract and also highly unstable and not depended on. This means that its abstractness is useless.
\item \textbf{The main sequence.} This is the ideal point for a package. A package near the main sequence is a package that conforms to the SAP and is as abstract as it is stable. The sequence is ideal and thus not many packages can truly be placed on this line, however the distance of a package from this ideal line can be measured. 
\end{itemize}

\begin{equation}
\label{eq:abstractness}
A = \frac{N_a}{N_c}
\end{equation}

\begin{equation}
\label{eq:abstract_stable_distance}
D = \frac{|A + I - 1|}{\sqrt{2}}
D' = |A + I - 1|
\end{equation}

\paragraph{}
In Eq. \ref{eq:abstract_stable_distance}, $D$ is the distance from the main sequence and $D'$ is its normalized version that ranges from $[0, 1]$.

\chapter{Community detection and applications}
\pagebreak
\paragraph{}
Community detection is defined as the process of finding communities of nodes in networks, such that the nodes inside a community have a higher property resemblance to one another compared to nodes in another community. A network is a graph with a pair of sets, $V$ and $E$, where $V$ is the set of all vertices and $E$ is the set of all edges in the network. Every community in a network is considered as a partition of the set $V$.
Typically, the process of detecting communities in a network consists of the following steps.

\begin{enumerate}
\item Specifying a quality measure that defines the quality of a partitioned graph. 
\item Using a specific method to assign nodes to different communities (clusters) in a way that increases the quality measure in step 1. 
\end{enumerate}

There are normally three different terms, related to the subject of community detection that are sometimes used interchangeably by mistake, thereforee distinguishing between these terms is needed before going into the details of each method. 

\begin{itemize}
\item \textbf{Graph partitioning vs Community detection.} The most important difference is that the problem of graph partitioning is universally defined as a problem where the number and sizes of the clusters are specified a priori. This is not the case in graph clustering or cluster analysis in general. The second, less important difference between these two terms is that clustering excludes the possibility of overlap by convention, so that it is still possible to speak of an overlapping clustering, whereas a partition or partitioning excludes the possibility of overlap by definition. 

\item \textbf{Graph clustering vs Community detection.} Graph clustering and community detection are normally used interchangeably in the litrature and in this thesis.

\end{itemize}

\begin{figure}[ht!]
\centering
\includegraphics[scale=0.3]{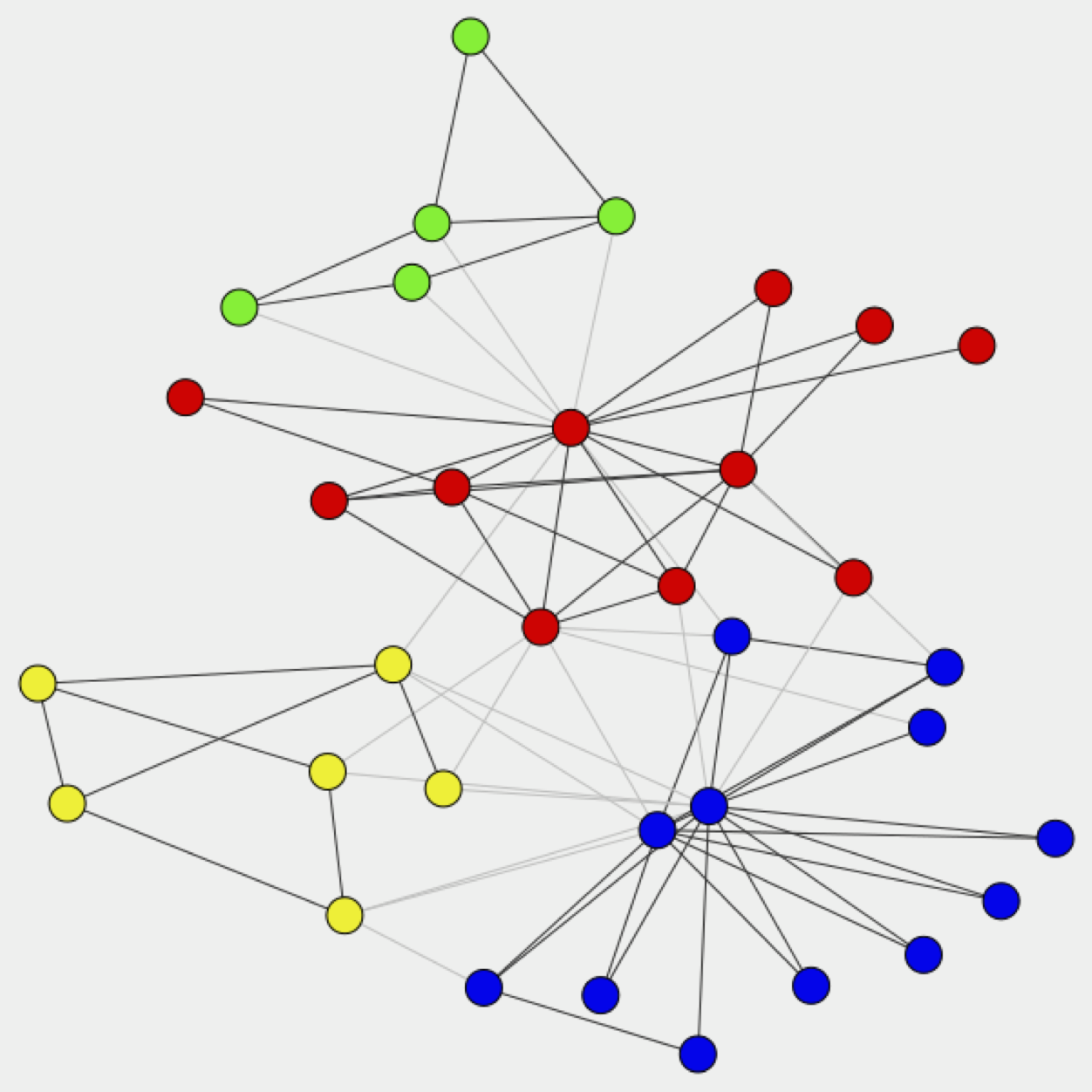}
\caption{The Zachary club}
\label{zachary_club}
\end{figure}

\paragraph{}
Many different community detection and graph partitioning algorithms have been proposed in the literature, some of which will be briefly discussed in this thesis. 

\section{Classification of clustering methods}
\paragraph{}
Graph clustering methods are normally difficult to classify, however Wiggerts \cite{wiggerts1997using} believes that they can generally be divided into the following methods.

\begin{itemize}
\item \textbf{Hierarchical methods.} Hierarchical approaches are known as some of the early solutions to the problem. These methods provide a hierarchy of partitions like a tree, known as a dendrogram. A sample dendrogram is depicted in Fig. \ref{zachary_club_dendrogram}. Hierarchical methods are themselves divided into the two groups of agglomerative approaches and divisive approaches. In agglomerative approaches, the algorithm starts with placing every node inside a separate cluster. Then the algorithm starts merging the clusters based on their similarity. It is important to note that the algorithm will not stop unless told to, thereforee knowing the number of wanted partitions in the result is crucial. In divisive hierarchical approaches, the algorithm starts with a single cluster that contains all the nodes of the graph. The algorithm then splits the cluster based on the similarity between the nodes, keeping the similar ones in the same cluster. 
Different hierarchical algorithms are distinguished by their distance function which is responsible for determining the similarity between two given nodes.

\item \textbf{Optimization based methods.} These algorithms generally take an initial inaccurate clustering and with the help of a quality measure, try to enhance and improve the cluster and optimize the quality. One of the most common and famous quality measures in the literature is the modularity measure proposed by Girvan and Newman \cite{chen2013new}. Various kinds of optimization techniques are applicable in this category of graph clustering algorithms, such as genetic algorithm based optimization methods, particle swarm methods, etc. A simple genetic algorithm approach can be like the following \cite{shtern2012clustering}.

\begin{enumerate}
\item Select a random population of partitions
\item Generate a new population by selecting the best according to a quality measure, such as Newman's modularity
\item Repeating step 2 until a certain criteria is met
\end{enumerate}

\item \textbf{Graph theoretical based methods.} Graph theoretical algorithms are methods that utilize the formal descriptions and properties of graphs and their respective subgraphs. In these methods, various subgraphs and properties are used to extract meaningful clusters from the original graph. Two important and common types of graph theoretical algorithms exist, namely aggregation algorithms and minimal spanning tree algorithms. Aggregation algorithms use the function of reduction on different nodes and merge them in each step. Different potential nodes for merging are chosen using different techniques, such as neighbourhoodness, strong connections and etc. Minimal spanning tree algorithms use the minimal spanning tree of the graph. These algorithms are normally not considered accurate as they tend to create large clusters, however some enhanced versions of these algorithms have been suggested in the literature \cite{shtern2012clustering}.

\item \textbf{Construction algorithms.} These algorithms assign nodes into clusters in one pass. The bisection algorithm and density search techniques are considered as examples of such methods. 

\end{itemize}

\begin{figure}[ht!]
\centering
\includegraphics[scale=0.3]{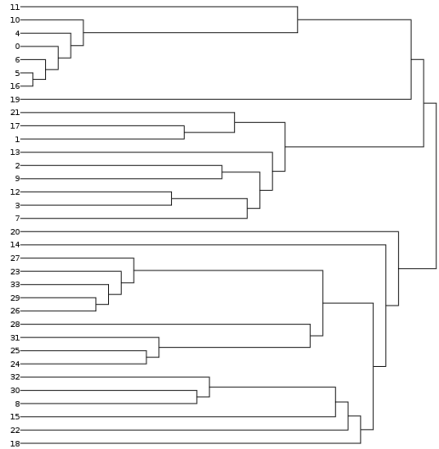}
\caption{A dendrogram for the Zachary club}
\label{zachary_club_dendrogram}
\end{figure}

\paragraph{}
The minimum cut approach is the most obvious and the most easiest way of tackling the problem of community detection. In this method, one tries to find two groups/partitions in a graph for which the edges connecting the two is the least. This approach mostly falls in the area of graph partitioning, because the number of partitions in the end result must be known a priori so that one can know how many times the algorithm should be applied. It is worth noting that if the minimum cut approach were to be used with no constraint, then a trivial solution to the problem would be to place all vertices in one partition only, thus minimizing the number of edges between partitions. Clearly this solution would not give any information on the \textit{communities} in a network. In the software engineering sense, the result of such a method would be a system with zero coupling and maximum cohesion, which seems the goal. However many important aspects of the software such as reusability, separation of concerns, object orientedness, flexibility, etc. will be lost. This raises the idea that maybe another measurement apart from coupling and cohesion is needed that can help find an optimum point for the two. This measure must be able to truly model and represent different objects in a software dependency network. In the subject of graph theory, a measure that can model the \textit{goodness} of a partition is known as a quality measure. Using a community quality measure in the field of software engineering has only recently been discussed in the literature \cite{pan2009class}, \cite{pan2013refactoring}. 

\section{Quality measures}
\paragraph{}
The quality of a partition found by a community detection algorithm is determined with a quality measure. This measure should show how good a partition is. Many algorithms provide many partitions without equal goodness, therefore it is absolutely necessary to measure the quality of the provided partitions and detect the best. Quality functions give a number to each partition so that the partitions can be ranked and compared to one another. Arguably, the most common and famous quality function is Newman and Girvan's Modularity \cite{newman2004finding}. 

\paragraph{}
Modularity is based on the idea that a random graph contains no meaningful community. Based on this idea, if one can make a similar graph to the one being analyzed with the same number of vertices, edges and degrees but with edges placed at random, then by comparing it to the original graph one can find the major differences that have created communities. To understand the notion of modularity, we start by another measure for the goodness of a partition and build on it. Let $G$ be a graph with elements of its adjacency matrix presented as $A_{vw}$, where $A_{vw}$ is 1 if nodes $v$ and $w$ are connected and 0 otherwise, and $C_v$ being the community in which vertex $v$ belongs to. The following measure shows the fraction of edges in graph $G$, that fall within communities.

\begin{equation}
\label{eq:Fraction of edges in the same community}
\frac{\sum_{vw} A_{vw}\delta(C_v, C_w) } {\sum_{vw}(A_{vw}) } = \frac{1}{2m} \sum_{vw}A_{vw}\delta(C_v, C_w)
\end{equation}

where $\delta$ is the Kronecker delta function and $m$ is the number of edges in the graph. 

\paragraph{}
This fraction takes the value of 1 when all edges fall in one community and hence is not a good enough measure.

\paragraph{}
The idea behind modularity is that a random graph does not have a meaningful community structure and thus, if generated carefully, should provide a good point of comparison. Carefully generating a random graph that can depict the features and properties of the original graph but with no meaningful community is known as providing a null model in the area of complex systems. In this case, one can provide a graph which has the same amount of vertices, edges and vertex degrees while its edges are rewired randomly, so that the graph looses its community structure. In such a graph, the probability of an edge being in between vertices $v$ and $w$, if connections are made at random is calculated as below.

\begin{equation}
\label{eq:The probability of an edge between v and w}
\frac{k_vk_w}{2m}
\end{equation}

where $k_v$ and $k_w$ are the degrees of vertex $v$ and $w$ respectively. Now, by using equations \ref{eq:Fraction of edges in the same community} and \ref{eq:The probability of an edge between v and w}, one can calculate the modularity measure as

\begin{equation}
\label{eq:Modularity}
Q=\frac{1}{2m}\sum_{vw}[A_{vw} - \frac{k_vk_w}{2m}] \delta(C_v, C_w).
\end{equation}

\paragraph{}
By looking at Eq. \ref{eq:Modularity}, one can see some important aspects of this measure. The Kronecker delta function makes sure that a connection between two graph nodes in two different communities makes no contribution to modularity. Two connected nodes inside a community, make a positive contribution to modularity and the contribution is inversely proportional to the degrees of the two nodes. Also two nodes that are not connected, yet still reside in one community provide a negative contribution to the overall modularity of the clustering. 

\section{A brief discussion of well known clustering methods}
In this section, several common graph clustering methods are briefly studied.
\subsection{The fast greedy method}
\paragraph{}
A typical greedy method for clustering a graph while utilizing Newman's modularity consists of the following steps.
\begin{enumerate}
\item Start with each vertex in its own community, thus having $n$ communities for $n$ vertices.
\item In each step, merge two communities whose join makes the highest increase in modularity $Q$.
\item After $n-1$ joins, one community remains and a dendrogram can be created. 
\item Take the clustered solution that has the highest Q. 
\end{enumerate}

\paragraph{}
The simple greedy method, can waste a good deal of time when dealing with sparse graphs. In the implementation of the simple greedy approach, one has to merge many columns and rows of the sparse adjacency matrix and consequently time and space is wasted on merging elements with the value of 0. 
For this reason, Clauset and Newman have presented an enhanced version of the greedy method, namely the fast greedy method \cite{clauset2004finding} which performs much better than many other algorithms in the literature. In the fast greedy method, some data structures such as max heaps and balanced binary trees are used with some alterations in the algorithm that results in the runtime of $|V||E|log(|V|)$.

\subsection{The edge-betweenness based method}
\paragraph{}
The edge betweenness based method, proposed by Girvan and Newman \cite{girvan2002community} before presenting the modularity measure, is a graph clustering algorithm that focuses on the edges that are between communities in contrast to many other older algorithms that focus on the connections inside a community. Edge betweenness is described as the number of shortest paths between pairs of vertices that run along it. The algorithm for this method is as follows.

\begin{enumerate}
\item Calculate edge betweeenness for all edges
\item Remove the edge with the highest betweenness value
\item Recalculate edge betweenness for the rest of the edges
\item Repeat step 2 until no edges remain
\end{enumerate}

\paragraph{}
Calculating betweenness for all $m$ edges and $n$ vertices of a graph can be calculated using Newman's algorithm for betweenness \cite{girvan2002community} which can be calculated in time $O(mn)$. Edge betweenness has to be recalculated for every edge removal and thus the algorithm can work in time $|N||M|^2$.

\subsection{The walktrap based method}
\paragraph{}
The walktrap method is based on the notion of random walks \cite{pons2005computing}. The main idea behind the walktrap method is that random walks in a graph tend to get trapped in dense parts of the graph which could represent communities. In the walktrap method, a distance $r$ between communities is calculated based on the properties of random walks. After this step, typically an agglomerative algorithm is used to merge communities and create a dendrogram, much like other methods. This algorithm has a runtime of $|E||V|^2$.

\subsection{The leading eigenvector based method}
\paragraph{}
The leading eigenvector algorithm utilizes the eigenvalues of the modularity matrix. In this algorithm one determines the eigenvector corresponding to the most positive eigenvalue of the modularity matrix and divide the network into two groups according to the signs of the elements of this vector. 

\section{Community detection for directed graphs}
\paragraph{}
Community detection in directed networks is a difficult task \cite{malliaros2013clustering}. Various algorithms for community detection in undirected graphs have been presented in the literature, however methods for directed approaches have been less common. A comprehensive survey of community detection methods for directed graphs can be found in \cite{malliaros2013clustering} by Malliaros et al. They propose the following classification for community detection approaches in directed graph.

\begin{enumerate}
\item \textbf{Naive graph transformation approach.} In this method, directions are simply removed from the graph and undirected community detection techniques are applied. 
\item \textbf{Transformations maintaining directionality.} In this category of methods, the graph is transformed to an undirected version while directionality is maintained using other methods. The original graph can be tranformed to a unipartite weighted graph or a bipartite graph for this approach. An overview of such transformations is depicted in Fig. \ref{transformation}.
\item \textbf{Extending objective functions and methodologies in directed graphs.} Many objective functions and quality measures used in undirected graphs can be extended to directed versions, i.e. modularity, spectral clustering, page rank and random walk methods, local density clustering. 
\item \textbf{Alternative approaches.} Other methods that can not be placed in the first three categories also exist. Such as information theoretic approaches and blockmodeling approaches.  
\end{enumerate}

\begin{figure}[ht!]
\centering
\includegraphics[width=.8\columnwidth]{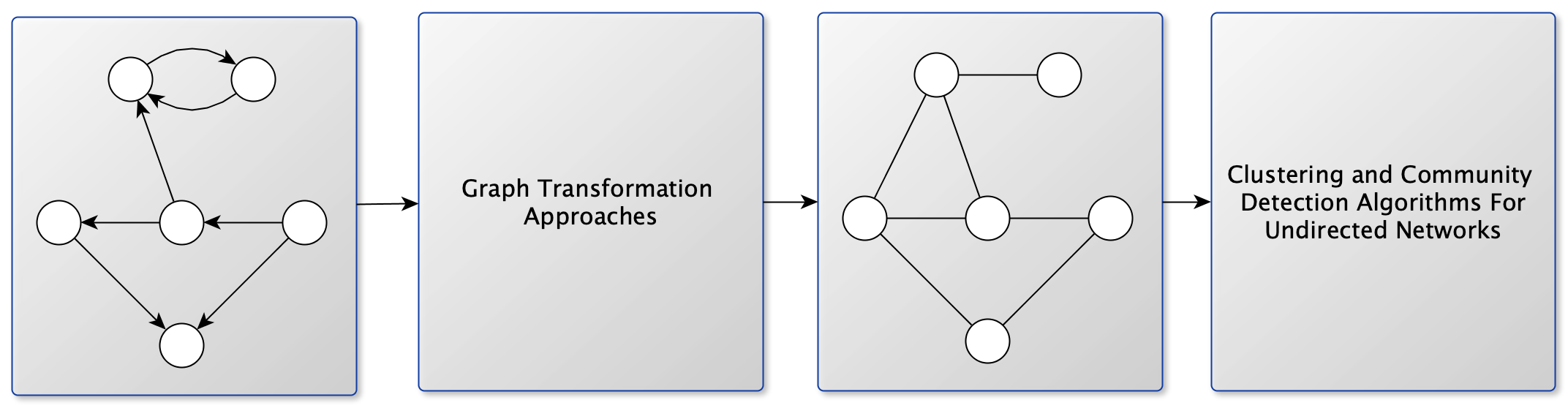}
\caption{An example of a transformation that preserves directionality. }
\label{transformation}
\end{figure}

\paragraph{}
Although some algorithms exist for this purpose, many clustering algorithms for undirected graphs can be extended for directed graphs with the help of a direction-compliant quality measure. Several extensions of modularity for directed graphs have been proposed in the literature. Arenas et al \cite{arenas2007size} proposed an extension of modularity. Their idea is based on the fact that in a directed graph $G$, if vertex $i$ exists with more out-links and vertex $j$ exists with more in-links, then it is more probable that in a random rewiring a link be found from $i$ to $j$ rather than the opposite. Considering the original idea of modularity, this suggests that if an edge is found from $j$ to $i$, then this edge is contributing to a community structure more than $i$ to $j$ would, simply because it is more suprising and less random. By this definition, modularity can be altered for directed networks by changing the null model to a graph with the same number of vertices, edges, out-links and in-links as the original graph. The equation for modularity $Q$ in a graph with the adjacency matrix $A$ and $m$ number of edges can then be expressed as

\begin{equation}
\label{eq:Modularity_directed}
Q=\frac{1}{m}\sum_{ij}[A_{ij} - \frac{k_i^{out}k_j^{in}}{m}] \delta(C_i, C_j)
\end{equation}

where $\delta$ is the Kronecker delta function, $C_i$ and $C_j$ denote the communities that nodes $i$ and $j$ belong to, and $k_i^{out}$ and $k_j^{in}$ are the number of vertex $i$ and $j$'s out-links and in-links respectively.

\section{Applications of community detection in software engineering}
\paragraph{}
Graph clustering is widely used in the literature as a method for finding meaning in a structure. This need for finding meaning in a complex system is generally used in four main areas of software engineering.

\subsection{Reflexion} 
\paragraph{}
Reflexion is the art of bridging the gap between software and humans, when it comes to analyzing a legacy system. Reflexion analysis tries to build an understandable high level abstraction of a large system, given the source code. In the process, the source code is analyzed and mapped to a new higher level model. This cumbersome task is typically done manually, however graph clustering can be used in semi automated mappings of source code to entities with the help the user’s knowledge about the system. Some related work has been presented in the literature \cite{murphy1995software}, \cite{mancoridis1999bunch}.

\subsection{Refactoring}
\paragraph{}
There are many properties that can be associated with good code. Sommerville describes good code as one that is highly maintainable, dependable, efficient and usable \cite{sommerville2004}. Truly reusable code is considered gold in the software industry as it significantly effects productivity and thus lowers costs \cite{lim1994} and without a doubt, good code is backed by a good design. Refactoring is the art of improving the internal structure of code while leaving the outer side intact \cite{fowler1999}. One of the problems that has been tackled in the literature is refactoring large and complicated legacy systems and also analyzing the structure of new code. Graph clustering techniques can be considered a good method for finding the correct structure and packages of a large system by analyzing the relationships in a software's dependency graph. Some work has been done in the area of refactoring at the class level, using graph clustering algorithms \cite{pan2009class}. Recently, some work has also been presented in the package level \cite{pan2013refactoring}, however the lack of an accurate package analysis tool that considers important object oriented aspects, such as stability and reusability is strongly felt in the literature. 

\subsection{Parallel computing}
\paragraph{}
Task to processor mappings is considered an important problem in parallel environments. The two general strategies used in such problems is placing tasks that can run concurrently on different processors, while keeping tasks that need many communications on the same processors, in order to increase locality. Graph partitioning tools have been used in some cases to map tasks to hypercube structures \cite{sadayappan1990cluster}.

\subsection{Ontologies and concept grouping}
\paragraph{}
One of the areas that highly utilizes graph clustering methods is ontologies and the semantic web. Various applications have been presented in the literature. One important application is extracting new concepts and taxonomies from ontologies. Extracting more generalized concepts and relations is one of the outputs of an ontology clustering. Tang et al. presents a great survey on such methods \cite{tang2012survey}. Modularization is also considered important for the problem of ever growing and over grown ontologies. The works in \cite{ghafourian2013modularization} is one of the most recent methods in this specific area. 

\section{Partition stability}
\paragraph{}
In some works the notion of partition stability, also known as robustness is considered as an important property of a good clustering algorithm. The idea is that a stable partition is one that can be recovered even if the structure of the graph is modified, as long as the change in the graph is not too extensive. It important to stress that this thesis only studies stability in the software package sense of the word and does not cover cluster stability. 

\chapter{Refactoring packages using community detection}
\pagebreak
\paragraph{}
Several studies have attempted to use community detection methods or cluster analysis in order to find refactoring opportunities \cite{pan2009class}, \cite{pan2013refactoring}, \cite{melton2006identifying}. These methods have analyzed the code in three main levels.

\begin{itemize}
\item Method level
\item Class level
\item Package level
\end{itemize}

\paragraph{}
This thesis focuses on refactoring at the package level for which there has been little discussions in the literature. Pan et al. \cite{pan2013refactoring} proposes a novel method for refactoring using the notion of modularity, however neglects the use of a directed clustering approach. In this chapter, the importance of a directed graph model is discussed with regard to the notion of package stability and an improved version of a package refactoring method using community detection is presented. 

\section{Basics of modeling packages with graphs}
\paragraph{}
As discussed in previous chapters, many metrics have been proposed for different software properties at the class level. At the package level, which is in a higher level in the abstraction hierarchy compared to a class, the most important property in the literature is the dependency between two packages. When a class inside a package depends on a class from another package, the former package is said to depend on the latter. 

\paragraph{}
Let $G$ be a graph with the adjacency matrix $A$. Vertices in $G$ represent classes and edge $E_{ij}$ between vertex $i$ and vertex $j$ resembles a dependency between the two classes. Communities in this graph represent package structure. A dependency between two classes can be any usage of methods or variables or inheritance. Classes are being modeled to graph vertices for the sole purpose of using community detection methods for finding appropriate clusters which represent packages and different relationships between classes are not considered different. 

\paragraph{}
A thorough metric for package dependencies has been proposed in \cite{gupta2009package} by Gupta et al, which takes into account the different types of connections between packages when sub-packages also exist in the software. The metric is validated using Briand's evaluation criteria \cite{briand1996property}. Gupta et al consider two classes of two packages connected if any of the following relationships are found between them. 

\begin{itemize}
\item Aggregation relationships between two classes, i.e., one class's attribute has the type of another class
\item Class inheritence or implementing interfaces
\item Method invocation of one class by the method of another class
\item A class's method referencing an attribute from another class
\item A class's method has a parameter of the type of another class
\item A class's method has a local variable of the type of another class
\item A class's method invoking a method having a parameter of the type of another class
\end{itemize}

\paragraph{}
By Gupta et al's metric, coupling between two packages $p_a^i$ and $p_b^i$, where $i$ denotes the hierarchical level, is expressed as

\begin{eqnarray*}
\label{eq:coup}
Coup(p_a^i, p_b^i) &=& \left\{\begin{matrix}
 0, & (n=0, m=0)\\ 
 \sum_{x=1}{n}\sum_{y=1}^{m}r(e_x^{i+1}, e_y^{i+1}) \\ 
 + \sum_{y=1}^{m}\sum_{x=1}^{n} r(e_y^{i+1}, e_x^{i+1}), & (n \geq 1 \wedge m \geq 1)
\end{matrix}\right.
\end{eqnarray*}

where $n$ and $m$ are the number of elements of package $p_a^i$ and $p_b^i$ respectively at hierarchy level $i+1$, and $r$ is the binary connection between elements. An example of different hierarchical levels given in \cite{gupta2009package} is depicted in Fig. \ref{hierarchy_levels}. The binary connection between elements ($r$) can be calculated as

\begin{eqnarray*}
\label{eq:r_gupta}
r(e_x^{i+1}, e_y^{i+1}) &=& \left\{\begin{matrix}
1, & e_x^{i+1} \rightarrow e_y^{i+1};\\ 
0, & otherwise;
\end{matrix}\right.
\end{eqnarray*}

where $e_x^{i+1} \rightarrow e_y^{i+1}$ denotes that element $x$ depends on element $y$.

\begin{figure}[ht!]
\centering
\includegraphics[width=.5\columnwidth]{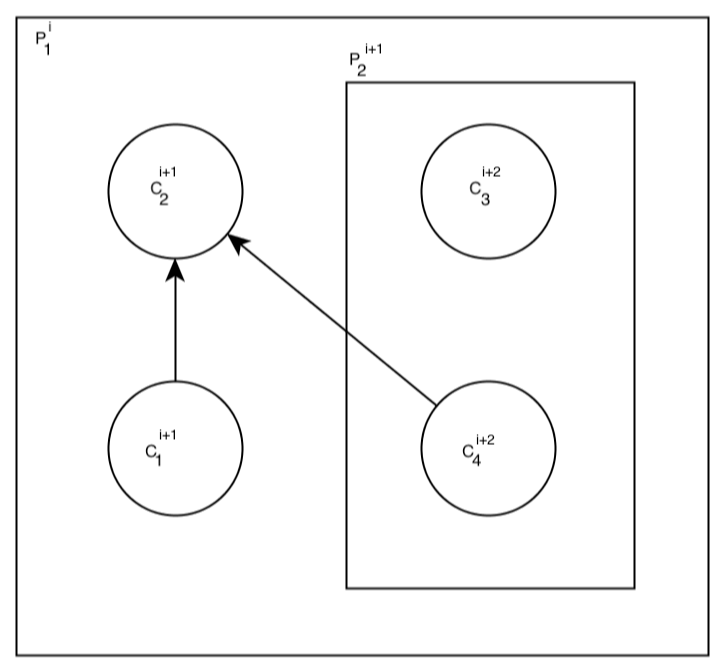}
\caption{An example of different hierarchical levels}
\label{hierarchy_levels}
\end{figure}

\section{Basics of refactoring with community detection}
\paragraph{}
The use of community detection methods for refactoring packages has only recently been studied in the literature by Pan et al \cite{pan2013refactoring}. An overview of their method is as follows.

\begin{enumerate}
\item Gather software information and dependencies from Java classes and jar files. 
\item Construct an undirected weighted dependency network based on the information gathered in the first step. 
\item Apply community detection to the dependency network to find the optimal placement of classes in packages. 
\item Compare the optimized clustering with the original packages structure of the code and suggest a list of possible refactoring candidates.
\end{enumerate}

\paragraph{}
In the first step of their algorithm Pan et al take into account two types of dependencies between code attributes, method accessing attribute dependency and method call dependency. Any of the two mentioned dependencies between two classes implies a dependency between the two classes. 

\paragraph{}
Pan et al model package structure with the help of two different networks, namely the undirected Feature Dependency Network (uFDN) and the undirected Weighted Class Dependency Network (uWCDN). Nodes in uFDN represent features inside the software and edges represent dependencies between features. By this definition, uFDN can be expressed as 

\begin{equation}
\label{eq:ufdn}
uFDN = (V_f, E_f, W_f)
\end{equation}

where $V_f$ and $E_f$ represent the set of vertices and edges in uFDN respectively and $W_f$ is the adjacency matrix for the network. The subscript $f$ shows that the two sets and the adjacency matrix are at the feature level. An example of a uFDN presented in \cite{pan2013refactoring}, consisting of two communities, is shown in Fig. \ref{ufdn}.

\begin{figure}[ht!]
\centering
\includegraphics[width=.7\columnwidth]{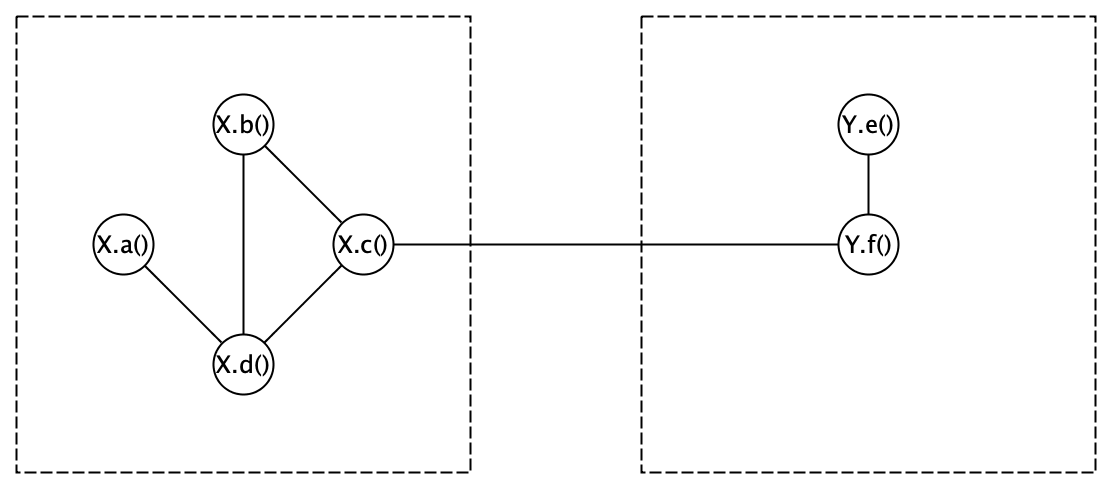}
\caption{A sample uFDN}
\label{ufdn}
\end{figure}

The code resembeling the network in \ref{ufdn} is given below. 

\begin{lstlisting}[language=java]
public class X
{
	private int a;
	public void c() {}
	public void b() {c();}
	public void d() {a++; b(); c();}
}
public class Y
{
	public void f()
	{
		X x = new X();
		x.c();			
	}
	public void e() {f();}
}
\end{lstlisting} 

\paragraph{}
In uWCDN, only the relationship among the classes are shown. A weight is used for every class dependency that represents the number of connections between the the attributes and methods of the two classes involved in the relationship. uWCDN can be defined as

\begin{equation}
\label{eq:uwcdn}
uWCDN = (V_c, E_c, W_c)
\end{equation}

where $V_c$ denotes the set of all vertices at the class level, $E_c$ denotes the set of all edges and $W_c$ is the weighted adjacency matrix of the network. Every entry in $W_c$ can be shown as $w_c(i, j)$ which is the weight between the two elements $i$ and $j$ and is used to denote the strength of a dependency between nodes $i$ and $j$. This weight can be calculated as

\begin{equation}
\label{eq:uwcdn_weights}
w_c(i, j) = |\bigcup_{n_i\in F_i}R_{i1} \cap F_j|
\end{equation}

\paragraph{}
The difference between uFDN and uWCDN is shown in Fig. \ref{ufdn_uwcdn}.

\begin{figure}[ht!]
\centering
\includegraphics[width=0.6\columnwidth]{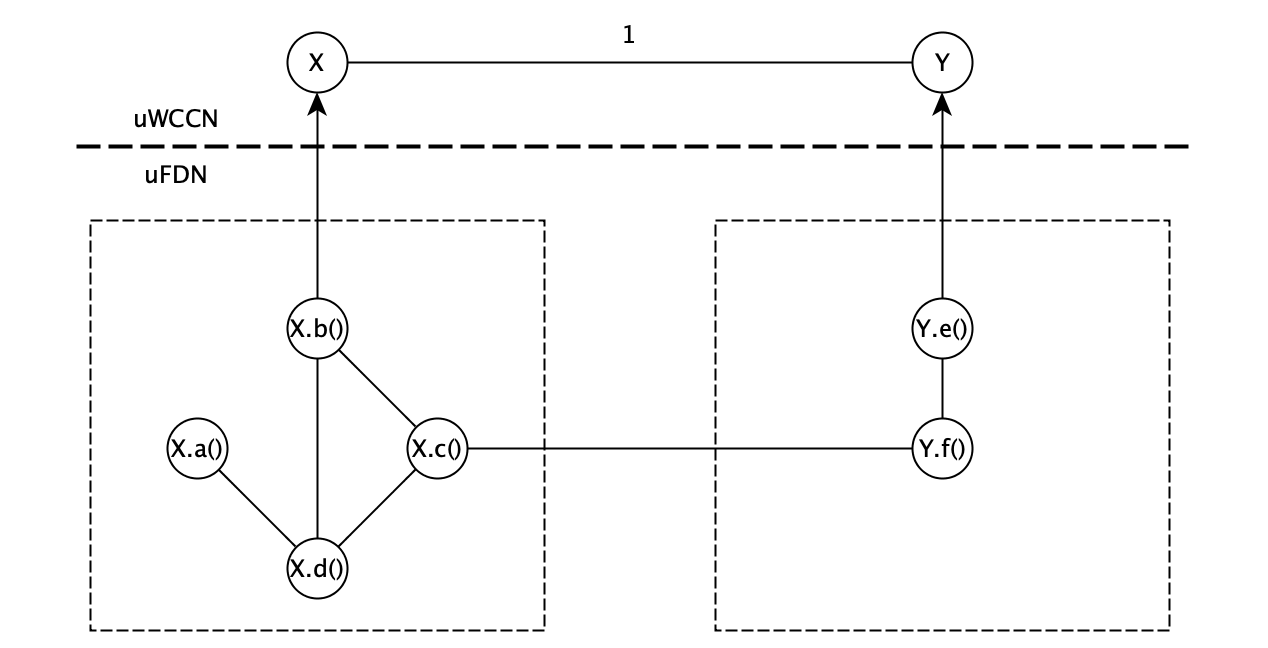}
\caption{A sample uWCDN compared to its respective uFDN}
\label{ufdn_uwcdn}
\end{figure}

where $R_{ik}$ denotes the set of all nodes reachable from $i$ within a distance of $k$ and $Fi$ is the set of all the features of class $i$. It is important to note that $w_c(i, j)$ is equal to $w_c(j, i)$.

\paragraph{}
The community detection algorithm used by Pan et al utilizes an older definition of modularity \cite{newman2004fast}.

\section{The importance of directed graphs in modeling package relationships}
\paragraph{}
Many studies in the literature have utilized undirected community detection methods for various applications. Fortunato \cite{fortunato2010community} presents a comprehensive review on undirected community detection methods. Many studies that include a directed model of a problem simply discard the information that the directions in the graph provides, and use a naive graph transformation approach. In the naive tranformation approach, graph directions are simply discarded and normal undirected community detection methods are applied to the graph. This can cause many important information to be discarded. We briefly discuss three main problems that an undirected approach can cause and how it effects refactoring and package stability. 

\subsection{Citation based cluster models}
\paragraph{}
Using naive transformation approaches for undirected community detection, introduces inaccuracy in certain graphs such as the citation based model that is depicted in Fig. \ref{citation_graph}. In this graph, the two middle vertices can clearly form a meaningful community. The two vertices have in-links from the the same set of vertices while the vertices that they have out-links to are also the same. In the package sense, the middle community resembles a package that is more stable than the package containing the vertices from the left. Many utility packages and libraries contain packages with a similar structure. There is little or no connection between the vertices inside the package, yet they belong to the same community as they are used in similar situations. 

\begin{figure}[ht!]
\centering
\includegraphics[width=0.2\columnwidth]{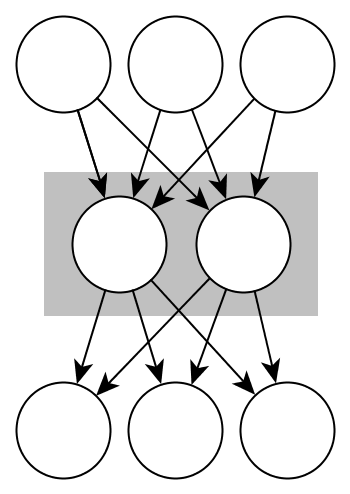}
\caption{Citation based cluster}
\label{citation_graph}
\end{figure}

\paragraph{}
After applying naive transformation and trying to find optimal communities in the graph in Fig. \ref{citation_graph}, the output simply looses the intended community structure. The output is given in Fig. \ref{citation_modular}. Black vertices have been put into one community by the algorithm and white vertices have been placed in another community. In this clustering, it is clear that SDP (Stable Dependencies Principle) is violated and both communities depend on each other. Using a community detection algorithm intended for undirected graphs has changed the SDP compliant structure that the programmer had intended. 

\begin{figure}[ht!]
\centering
\includegraphics[scale=0.2]{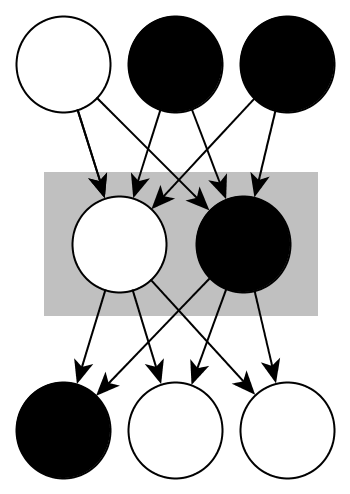}
\caption{Citation based cluster after naive transformation}
\label{citation_modular}
\end{figure}

\subsection{Bidirected graphs and loss of information}
\paragraph{}
As discussed in \cite{malliaros2013clustering}, the information needed for correct community detection is simply lost in certain graphs such as the bidirected graph shown in Fig. \ref{bidirected_graph}.

\begin{figure}[ht!]
\centering
\includegraphics[scale=0.6]{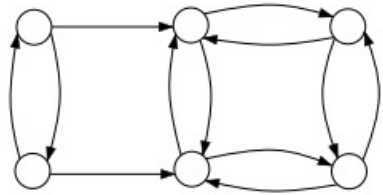}
\caption{An example of a bidirected graph with two communities}
\label{bidirected_graph}
\end{figure}

\begin{figure}[ht!]
\centering
\includegraphics[scale=0.6]{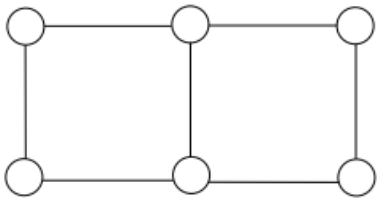}
\caption{An example of a bidirected graph after naive transformation}
\label{bidirected_no_direction}
\end{figure}

From a stability perspective, the dependency graph in Fig. \ref{bidirected_graph} shows a two packages that fully conform to SDP. The community created by the four vertices on the right represent a very stable package that the left community is depending upon. By performing naive transformation the graph would look like Fig. \ref{bidirected_no_direction}. This graph has lost its community structure and the two left most vertices and the two right most vertices will be treated in the same way when it is given to a community detection method. Fig. \ref{bidirected_no_direction_modular} shows this graph after applying community detection while optimizing Newman's modularity. 

\begin{figure}[ht!]
\centering
\includegraphics[scale=0.6]{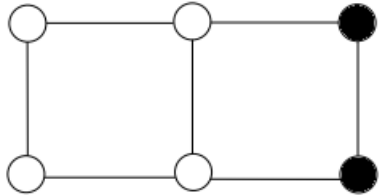}
\caption{A clustered version of the graph in Fig \ref{bidirected_no_direction}}
\label{bidirected_no_direction_modular}
\end{figure}

\section{Stability and modularity}
In this section, the relationship between the directed version of modularity and the Stability Dependencies Principle (SDP) in refactoring packages is discussed. In a scenrio where a class is chosen to be moved from one package to another using community detection methods, we show that modularity is in favor of SDP and hiding dependencies that violate SDP inside packages has a higher contribution to modularity than hiding non-violating dependencies. To show this behavior, some prior definitions are needed. 

\begin{defi}
A movement of class $i$ from package $p_1$ to package $p_2$ is shown as the tuple $(i, p_1, p_2)$.
\end{defi}

\begin{defi}
A border node in a package is defined as a node that has connections with nodes in other packages and thus directly effects the package's instability metric.
\end{defi}

\paragraph{}
SDP is generally satisfied in a case where no stable package depends on an unstable package. When considering the movement of only two border classes, while all other classes and packages are left intact, then the only dependencies effecting the two package's instability metric are the dependencies of the two border nodes. If a border node $i$ from stable package $p_1$ depends on a node $j$ from unstable package $p_2$, then clearly SDP is violated. 

\begin{remark}
\label{sdp_satisfaction}
Let $k_i^{out}$ and $k_j^{out}$ be the out-link degree of vertices $i$ and $j$ respectively, and $k_i^{in}$ and $k_j^{in}$ be the in-link degree of vertices $i$ and $j$. 
If $k_i^{out} > k_i^{in}$ and $k_j^{out} < k_j^{in}$ and node $i$ and node $j$ are border nodes, then SDP is satisfied.
\end{remark}

\begin{remark}
\label{sdp_dissatisfaction}
Let $k_i^{out}$ and $k_j^{out}$ be the out-link degree of vertices $i$ and $j$ respectively, and $k_i^{in}$ and $k_j^{in}$ be the in-link degree of vertices $i$ and $j$. 
If $k_i^{out} < k_i^{in}$ and $k_j^{out} > k_j^{in}$ and node $i$ and node $j$ are border nodes, then SDP is not satisfied.
\end{remark}

\begin{prop}
\label{prop:sdp_mod}
Let $i$ and $j$ be two classes in dependency graph $G$. If a movement $(i, c_i, c_j)$ exists and the conditions of remark \ref{sdp_satisfaction} holds, then the increase in modularity $Q$ is more, compared to the situation in which the conditions of remark \ref{sdp_dissatisfaction} holds true. 
\end{prop}

\begin{proof}
Let $Q$ denote modularity while the conditions in remark \ref{sdp_satisfaction} holds true and $\bar{Q}$ denote modularity while the conditions in remark \ref{sdp_dissatisfaction} holds true. $Q$ and $\bar{Q}$ can be calculated using Eq. \ref{eq:Modularity_directed} as

\begin{eqnarray*}
Q&=&\frac{1}{m}\sum_{ij}[A_{ij} - \frac{k_i^{out}k_j^{in}}{m}] \delta(C_i, C_j), \\
\bar{Q}&=&\frac{1}{m}\sum_{ij}[A_{ij} - \frac{\bar{k}_i^{out}\bar{k}_j^{in}}{m}] \delta(C_i, C_j).
\end{eqnarray*}

The bar on in-link or out-link $k$ denotes that it is being calculated in the scenario of remark \ref{sdp_dissatisfaction}, and is therefore equivelant to the out-link and in-link in the scenario of remark \ref{sdp_satisfaction} respectively. Thus one can write

\begin{eqnarray*}
\label{eq:kij_bar_relation}
\bar{k}_i^{out} &=& k_j^{out} \\
\bar{k}_j^{in} &=& k_i^{in}.
\end{eqnarray*}

By looking at the conditions of remark \ref{sdp_satisfaction} and remark \ref{sdp_dissatisfaction} it is clear that

\begin{eqnarray*}
\label{eq:kij_bar_relation2}
\bar{k}_i^{out}\bar{k}_j^{in} &<& k_i^{out}k_j^{in} \\
\frac{\bar{k}_i^{out}\bar{k}_j^{in}}{m} &<& \frac{k_i^{out}k_j^{in}}{m} \\
A_{ij} - \frac{\bar{k}_i^{out}\bar{k}_j^{in}}{m} &>& A_{ij} - \frac{k_i^{out}k_j^{in}}{m} \\
\bar{Q} &>& Q.
\end{eqnarray*}

\end{proof}

\paragraph{}
The above proposition shows how modularity is compatible with the notion of SDP. Modularity is in favor of non-random structure in a network. Violating SDP would mean that a stable package is depending on an unstable package. In this scenario, the above proof shows that keeping two nodes that have violated SDP before, inside a single package is better for $Q$ than keeping two nodes that did not violate SDP. It is also important to note that if $i$ and $j$ belong to two different packages, then the condition will have no contribution to modularity and therefore is not discussed. 

\paragraph{}
As an example for the proved proposition, suppose that a system contains two packages $C_1$ and $C_2$, where $C_1$ is unstable and $C_2$ is a stable package. Two slighly different versions of this system is depicted in Fig. \ref{example1}. In both of these versions, vertices 1, 2, 3 and 4 are members of $C_2$ and vertices 5, 6, 7 and 8 belong to $C_1$. It is clear that in condition (b), edge $(1,5)$ is violating SDP. Based on Proposition \ref{prop:sdp_mod}, we show that moving node 1 from $C_2$ to $C_1$ has more positive contribution for package modularity, than in the case of condition (a). If movement $(1, C_2, C_1)$ happens, then four new edges positively contribute to the overall modularity of the dependency graph while one edge's contribution is eliminated. The reason for this is that edges between two communities provide no contribution to modularity because the kronecker delta function in Eq. \ref{eq:Modularity_directed} becomes zero. therefore edges $(5, 1)$, $(6, 1)$, $(7, 1)$ and $(8, 1)$ will have new contributions to modularity and edge $(1, 3)$ will no longer have any contribution. The changes in modularity $Q$ for condition (b) can be calculated using Eq. \ref{eq:Modularity_directed} as

\begin{eqnarray*}
\label{eq:example1}
\Delta Q &=& \overbrace{4(1 - \frac{1 \times 1}{2m})}^{\mathclap{\text{Contribution of the 4 new edges}}} - \underbrace{(1 - \frac{1 \times 1}{2m})}_{\mathclap{\text{Contribution of edge }(1, 3)}} = 3(1 - \frac{1 \times 1}{2m}).
\end{eqnarray*}

\paragraph{}
By replacing $m$ with the number of edges, we have

\begin{eqnarray*}
\label{eq:example1b}
\Delta Q &=& \frac{57}{20} = 2.85.
\end{eqnarray*}

\paragraph{}
Changes in modularity for condition (a) can be calculated the same way as follows.

\begin{eqnarray*}
\label{eq:example1c}
\Delta Q &=& \overbrace{(1 - \frac{4 \times 4}{2m})}^{\mathclap{\text{Contribution of edge }(5, 1)}} + 3(1-\frac{1 \times 4}{2m}) - (1 - \frac{1 \times 1}{2m}) = \frac{33}{20} = 1.65.
\end{eqnarray*}

\paragraph{}
The results clearly indicate that the graph gained more modularity when trying to suppress an SDP violation than when it is not. 

\begin{figure}[ht!]
\centering
\includegraphics[scale=0.5]{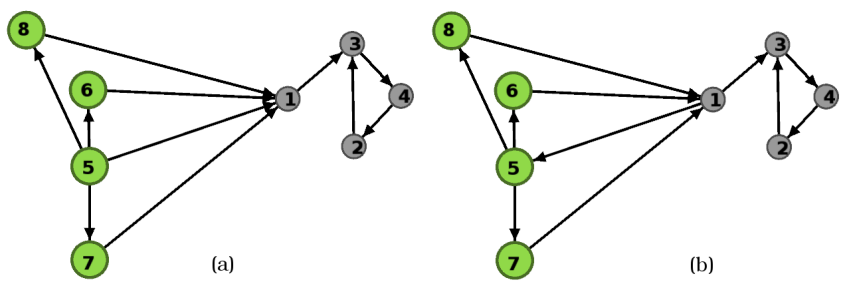}
\caption{Two different graph dependency conditions.}
\label{example1}
\end{figure}

\section{Proposed refactoring method}
\paragraph{}
By considering the discussed importance of directed graphs in refactoring software packages and the package coupling metric proposed by \cite{gupta2009package}, we present a package refactoring algorithm. 

\paragraph{}
For calculating the dependencies, we use the package coupling metric provided by Gupta et al \cite{gupta2009package} at hierarchy level $i+1$. This is a crucial point that must be noted. Hierarchy level $i+1$ is being used because it gives access to elements inside packages at level $i$. The classes and sub-pakages in this level of hierarchy are the ones that will be analyzed for possible refactorings. In this study, only one package level is analyzed for refactoring, as deeper levels cause many open problems that need to be tackled. The most basic problem with optimizing software metrics such as coupling and cohesion in many levels of abstractness simultaneously is that cohesion inside one level can be considered as coupling in a deeper level, thus the problem of minimizing coupling contradicts with the problem of maximizing cohesion in a higher level of abstractness, i.e., the package level $i$. therefore, in this work, only packages at level $i$ and their respective elements at level $i+1$ are considered. 

\paragraph{}
For calculating the package dependency graph's modularity, we use the directed and weighted version of modularity \cite{arenas2007size} expressed as

\begin{eqnarray*}
\label{eq:directed_weighted_mod}
Q&=&\frac{1}{2w}\sum_{ij}[w_{ij} - \frac{w_i^{out}w_j^{in}}{2w}] \delta(C_i, C_j). \\
\end{eqnarray*}

where $w_i^{out}$ and $w_j^{in}$ are respectively the output and input weights of nodes $i$ and $j$ and 

\begin{eqnarray*}
\label{eq:wiout}
w_i^{out} &=& \sum_{j}w_{ij} \\
w_j^{in} &=& \sum_{i}w_{ij} \\
2w &=& \sum_{i} w_i^{out} = \sum_{j} w_j^{in} = \sum_{i} \sum_{j} w_{ij}.
\end{eqnarray*}

The weights for an edge is equal to the edge's coupling metric given in Eq. \ref{eq:r_gupta}. These weights are used in the package dependency network, similar to the weights in uWCDN (Eq. \ref{eq:uwcdn}) provided by Pan et al \cite{pan2013refactoring}. Considering the directedness of the network we can define an enhanced version of uWCDN, namely DWPDN (Directed, Weighted Package Dependency Network) that can be expressed as

\begin{equation}
\label{eq:dwpdn}
DWPDN = (V_{i+1}, E_{i+1}, W_{i+1})
\end{equation}

where $V_{i+1}$ denotes the set of all vertices at hierarchy level $i+1$, $E_{i+1}$ denotes the set of all edges at hierarchy level $i+1$ and $W_{i+1}$ is the assymetric and weighted adjacency matrix of the network at hierarchy level $i+1$. Every element of $W_{i+1}$ can calculated as 

\begin{equation}
\label{eq:dwpdn_weights}
w_{i+1}(j, k) = Coup(p_j^{i+1}, p_k^{i+1})
\end{equation}

where $j$ and $k$ are two elements and $Coup$ is the coupling function from Eq. \ref{eq:coup}. 

\paragraph{}
The main phases of the proposed package refactoring algorithm are presented in Alg. \ref{alg:proposed_refactoring}. 

\begin{algorithm}[htbp]
\caption{Proposed refactoring algorithm}\label{alg:proposed_refactoring}
\textbf{Input: }A DWPDN\\
\textbf{Output: }A list of package movement suggestions and the optimal $Q$ that can be gained
\begin{algorithmic}[1]	
\Procedure{}{}
\State $suggestedMovements \gets \varnothing$
\State $Q' \gets -1$
\State $Q \gets \text{modularity based on Eq. \ref{eq:directed_weighted_mod}}$
\State $selectedCommunity \gets 0$
\For {every node $i$}
	\State $C_i \gets $ node $i$'s community
	\For {every node $j$}
		\State $C_j \gets $ node $j$'s community
		\State $tempQ \gets \text{modularity, while considering node } i \text{ in }C_j$
		\If {$tempQ > Q'$}
			\State $Q' \gets tempQ$
			\State $selectedCommunity \gets C_j$
		\EndIf
	\EndFor
	\If {$Q' > Q$}
		\State Add movement $(i, C_i, selectedCommunity)$ to suggestedMovements
		\State Move node $i$ to $selectedCommunity$
		\State $Q \gets Q'$
		\State $i \gets 1$
	\EndIf
\EndFor
\State \textbf{return} $Q'$ and $suggestedMovements$
\EndProcedure
\end{algorithmic}
\end{algorithm}

\chapter{Evaluation}
\pagebreak

\paragraph{}
This chapter evaluates the proposed algorithm with two case studies using our implemented package refactoring tool. The two subjects which the algorithm was applied on are the same open source subjects used in \cite{pan2013refactoring}. In the remaining sections of this chapter, the two subjects are briefly introduced and analyzed by the implemented tool. The proposed refactoring algorithm is applied on the two subjects and the results are evaluated. For simplicity, the first version of the implemented tool does not apply weights and considers all weights between classes to be one. 

\section{Subjects}
\paragraph{}
The two subjects being analyzed in this chapter are the same as the subjects in \cite{pan2013refactoring}, namely Trama\footnote{\url{http://trama.sourceforge.net}} and FrontEndForMySQL\footnote{\url{http://frontend4mysql.sourceforge.net}}. 

\paragraph{}
Trama is a graphical tool for manipulating and working with matrices. FrontEndForMySQL is a graphical front end for the MySQL database system and provides an easier and more user friendly environment for working with MySQL queries. Some details of the two subjects are shown in Table \ref{subjects}.

\begin{table}[h!]
	\caption{Details of the systems analyzed}
	\label{subjects}
	\centering
    \begin{tabular}{|l|l|l|l|l|}
    \hline
    \textbf{System} & Version  & Number of packages & Number of classes \\ \hline
    Trama        & 1.0 & 6 & 58  \\ \hline
    FrontEndForMySQL        & 1.0 & 10 & 56  \\ \hline
    \end{tabular}
\end{table}

\paragraph{}
The original packaging structure for Trama is depicted in Fig. \ref{original_trama}. The original modularity calculated for the default packaging of Trama is calculated as 0.28 and the list of its packages is as follows.

\begin{itemize}
\item visao
\item visao.renderizador
\item persistencia
\item negocio
\item negocio.leitor.Interface
\item negocio.leitor
\end{itemize}

\begin{figure}[ht!]
\centering
\includegraphics[scale=0.6]{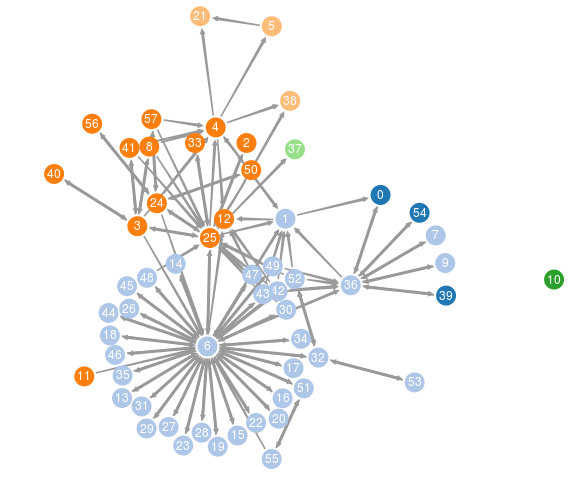}
\caption{Original packaging structure of Trama}
\label{original_trama}
\end{figure}

\paragraph{}
FrontEndForMySQL is a larger system compared to Trama, with an initial package modularity of 0.21. The system's default packaging structure is depicted in Fig. \ref{front_original} and it contains the following packages.
\begin{itemize}
\item frontendformysql
\item frontendformysql.domain.BackEnd
\item frontendformysql.domain.BackEndData
\item frontendformysql.domain.BackEndComponent.Editor
\item frontendformysql.domain.BackEndInterfaces
\item frontendformysql.domain.BackEnd.System
\item frontendformysql.domain.BackEndComponent.DriverModule
\item frontendformysql.domain.BackEndData
\item frontendformysql.domain.BackEndComponent.XMLutil
\item frontendformysql.domain.BackEndComponent.IO
\item frontendformysql.domain.BackEndComponent.DataStructures
\item frontendformysql.domain.BackEndComponent.Editor
\item frontendformysql.domain.BackEndInterfaces
\item frontendformysql.domain.BackEnd.System
\item frontendformysql.domain.BackEndComponent.DriverModule
\item frontendformysql.domain.BackEndComponent.XMLutil
\item frontendformysql.domain.BackEndComponent.IO
\item frontendformysql.domain.BackEndComponent.DataStructure
\end{itemize}

\begin{figure}[ht!]
\centering
\includegraphics[scale=0.6]{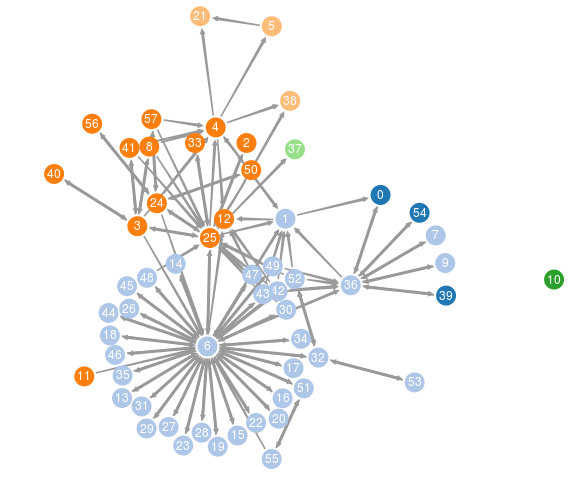}
\caption{Original packaging structure of FrontEndForMySQL}
\label{front_original}
\end{figure}

\section{Case studies and results}
\paragraph{}
After applying the proposed refactoring algorithm, with considering the importance of edge directions, the clustering of Trama changes to the depicted structure in Fig. \ref{refactor_trama} and the suggested movements are given in Table \ref{trama_suggestions}. The new packaging of Trama has a directed modularity of 0.43 and shows an improvement over the original 0.28. It is important to note that not all movements are acceptable and the suggestions should be given to a programmer for final analysis. 

\begin{table}[h!]
	\caption{Suggested movements for Trama classes}
	\label{trama_suggestions}
	\centering
    \begin{tabular}{|l|l|l|l|l|}
    \hline
    \textbf{Order} & \textbf{Class name} & \textbf{Old package}  & \textbf{Suggested package} \\ \hline
	1	&	Main		&	negocio	&	visao\\ \hline
	2	&	Matriz		&	negocio	&	persistencia\\ \hline
	3	&	ModeloTabela		&	visao	&	persistencia\\ \hline
	4	&	JTableCustomizado	&	visao	&	visao.renderizador\\ \hline
	5	&	JTableCustomizado\$1	&	visao	&	visao.renderizador\\ \hline	
	6	&	JTableCustomizado\$2	&	visao	&	visao.renderizador\\ \hline		
	7	&	LeitorDeModelo	&	negocio.leitor	&	negocio\\ \hline
	8	&	Tela\$23		&	visao	&	persistencia\\ \hline
	9	&	Tela\$22		&	visao	&	persistencia\\ \hline
	10	&	Tela\$24		&	visao	&	visao.renderizador\\ \hline
	11	&	Tela\$3\$1	&	visao	&	visao.renderizador\\ \hline
    \end{tabular}
\end{table}

\begin{figure}[ht!]
\centering
\includegraphics[scale=0.6]{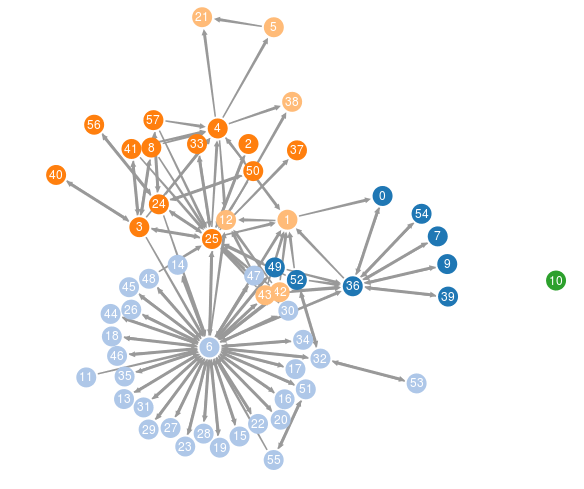}
\caption{New packaging of the Trama system after refactoring}
\label{refactor_trama}
\end{figure}

\paragraph{}
As a comparison, an undirected version of the algorithm, using naive transformation, was applied on the Trama system. The produced clustering is shown in Fig. \ref{trama_undirected}. In this clustering, modularity gets a value of 0.41. It is important to note that comparing the modularity of the two approaches would not be correct, as the formula for the two quality measures are inherently different. However, a comparison on package instability is shown in Table \ref{trama_instability_comparison}, in which $OI$ is the original instability of a package, $DI$ is the instability of a package after the proposed refactoring algorithm with edge directions, is applied and $UI$ is the instability of a package after applying the undirected version of the algorithm. 

\begin{figure}[ht!]
\centering
\includegraphics[scale=0.6]{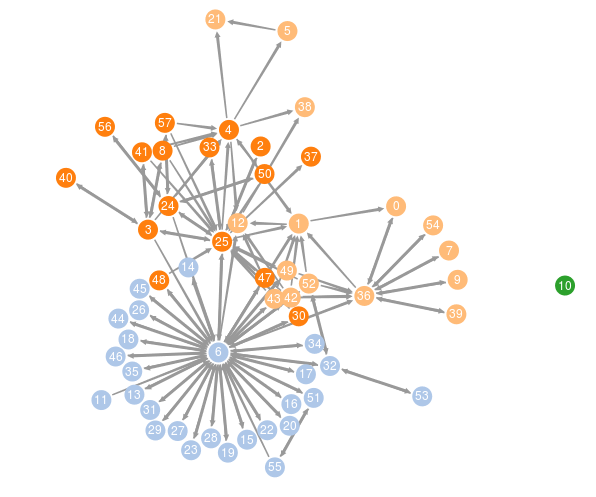}
\caption{New packaging of the Trama system after refactoring with naive transformation}
\label{trama_undirected}
\end{figure}

\begin{table}[h!]
	\caption{Comparison of Trama's instability metric for different approaches}
	\label{trama_instability_comparison}
	\centering
    \begin{tabular}{|l|l|l|l|l|}
    \hline
    \textbf{Package name} & \textbf{OI} & \textbf{DI}  & \textbf{UI} \\ \hline
	negocio						&	0.478 	&	0.529	&	0.6		\\ \hline
	persistencia				&	0		&	0.368	&	0.409	\\ \hline
	visao.renderizador			&	0.428	&	0.538	&	0		\\ \hline
	negocio.leitor				&	0		&	0		&	0		\\ \hline
	visao						&	0.64	&	0.578	&	0.5		\\ \hline	
	negacio.leitor.Intergface	&	0		&	0		&	0		\\ \hline		
    \end{tabular}
\end{table}

\paragraph{}
Table \ref{trama_instability_comparison} shows how two packages became more stable after applying the proposed, directed clustering algorithm,  while the stability of package \textit{visao} decreased by 0.078. From Fig. \ref{trama_undirected}, it is also clear that the \textit{visao.renderizador} is merged with other packages and thus is not taken into account for comparison.

\paragraph{}
The implementation of the proposed algorithms was also applied to the FrontEndForMySQL system. The original package structure for FrontEndForMySQL and its structure after refactoring are depicted in Fig. \ref{original_frontendformysql} and Fig. \ref{refactor_frontendformysql} respectively. The original modularity for FrontEndForMySQL is calculated as 0.21. 

\begin{figure}[ht!]
\centering
\includegraphics[scale=0.5]{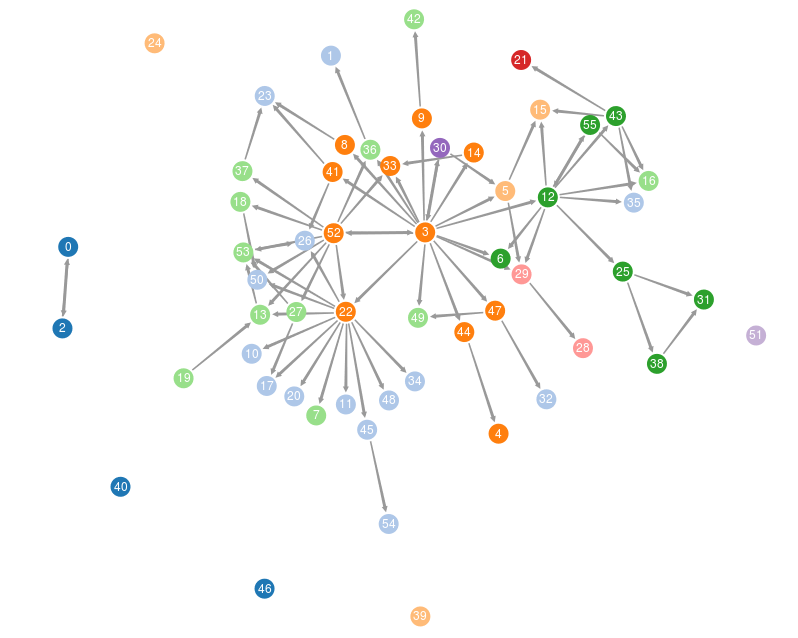}
\caption{Original packaging of the FrontEndForMySQL system}
\label{original_frontendformysql}
\end{figure}

\begin{figure}[ht!]
\centering
\includegraphics[scale=0.5]{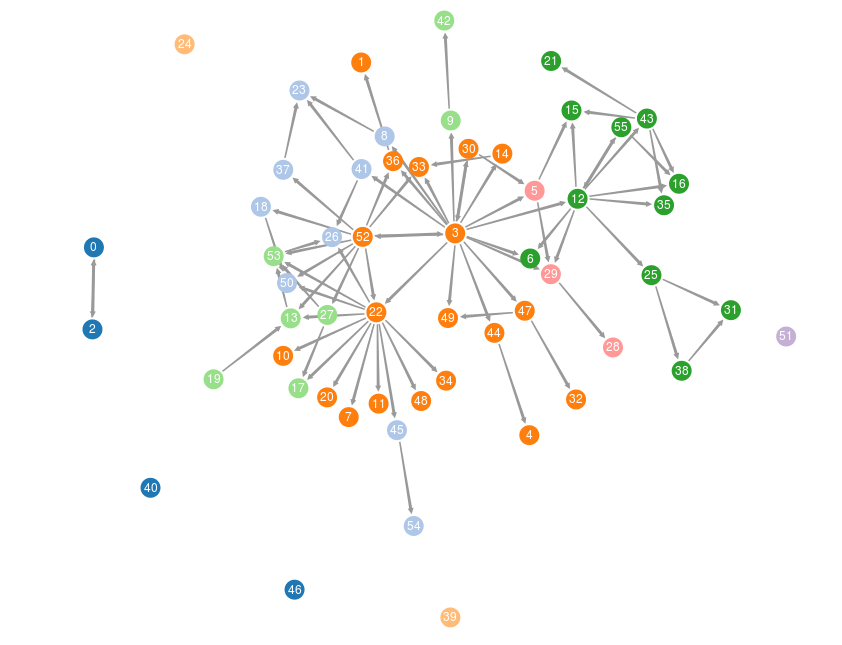}
\caption{New packaging of the FrontEndForMySQL system after refactoring}
\label{refactor_frontendformysql}
\end{figure}

\paragraph{}
Similar to the previous case study, an undirected version of the algorithm, using a naive transformation for removing edge directions was applied to FrontEndForMySQL and the clustering result is depicted in Fig. \ref{front_undirected}. The comparison table for this package instability measures is given in Table \ref{front_instability_comparison}.

\begin{figure}[ht!]
\centering
\includegraphics[scale=0.5]{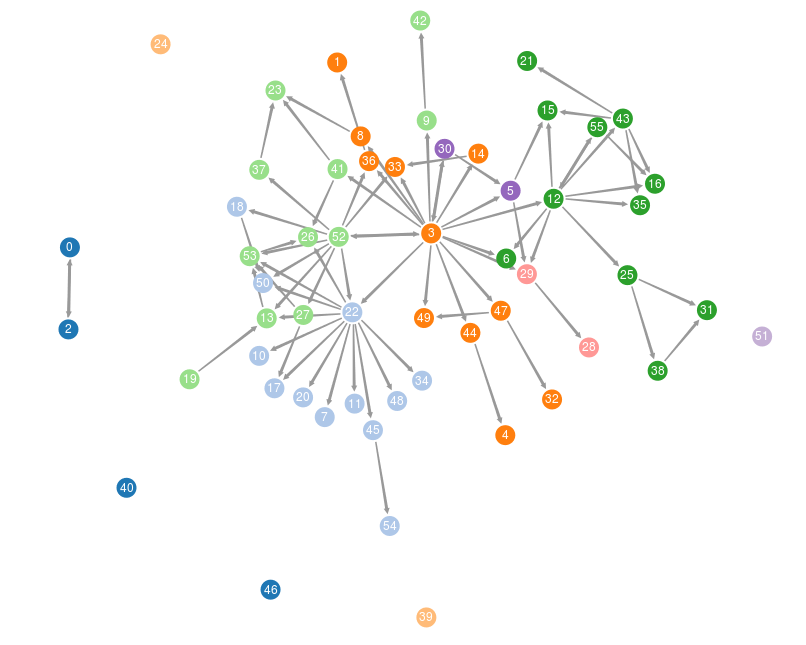}
\caption{New packaging of the FrontEndForMySQL system after refactoring with naive transformation}
\label{front_undirected}
\end{figure}

\begin{table}[h!]
	\caption{Comparison of FrontEndForMySQL's instability metric for different approaches}
	\label{front_instability_comparison}
	\centering
    \begin{tabular}{|l|l|l|l|l|}
    \hline
    \textbf{Package name} & \textbf{OI} & \textbf{DI}  & \textbf{UI} \\ \hline
	BackEndInterfaces						&	0 		&	0		&	0.375		\\ \hline
	BackEnd									&	0.969	&	1		&	0.714		\\ \hline
	BackEnd.System							&	0.2		&	0		&	0		\\ \hline		
	BackEndComponent.IO						&	0		&	0.2		&	0		\\ \hline
	BackEndComponent.XMLutil				&	0		&	0		&	0		\\ \hline	
	BackEndComponent.Editor					&	0		&	0		&	0		\\ \hline	
	BackEndComponent.DriverModule			&	0.818	&	0.25	&	0.25		\\ \hline		
	BackEndComponent.DataStructures		&	0		&	0		&	0	\\ \hline
	frontendformysql						&	0.666	&	0		&	0.6		\\ \hline		
	BackEndData								&	0.238	&	0.125	&	0.5		\\ \hline		
    \end{tabular}
\end{table}

\paragraph{}
Table \ref{front_instability_comparison} clearly shows that the overall instability of packages is higher when edge directions are not taken into account in the refactoring algorithm. 

\chapter{Live analysis of graph clusters}
\pagebreak
\paragraph{}
Considering the vast number of graph clustering applications in software engineering, a need for a tool that can import different graph modeled structures, perform graph clustering algorithms and provide a rich client for tweaking the properties of the model, is clearly felt. This need has motivated us to create a tool, namely Picasso, with such capabilities. Fig. \ref{picasso_overview} shows an overall view of this tool while analyzing a software package. Colors are used to show the different communities inside the graph. The server-side and client-side codes of this tool are given in Appendices A and B respectively.

\section{Picasso overview}
\begin{figure}[ht!]
\centering
\includegraphics[scale=0.35]{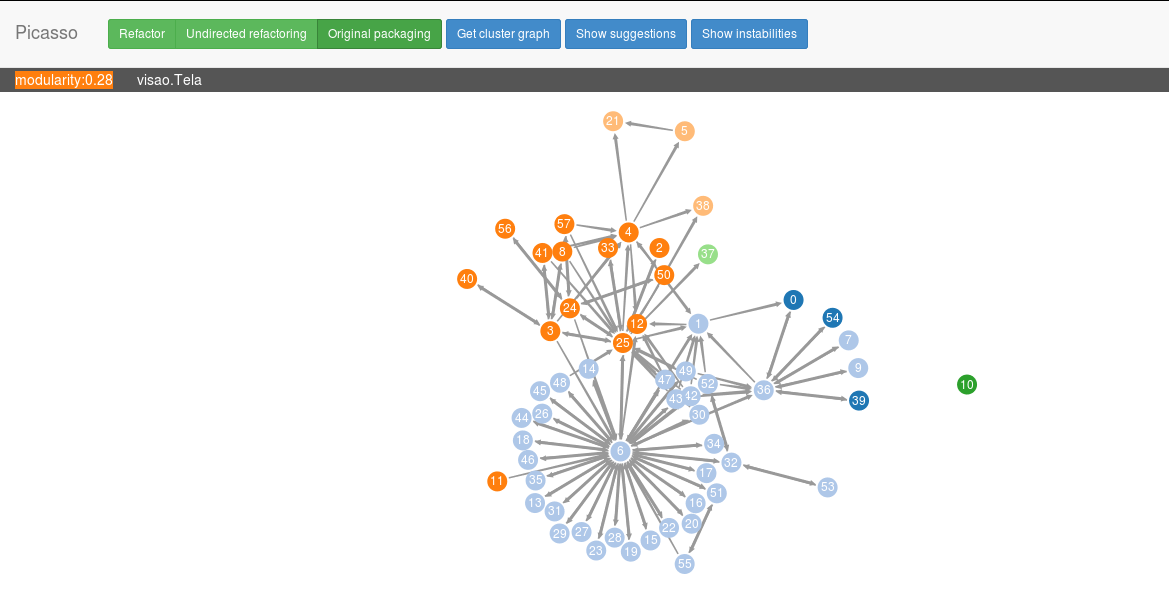}
\caption{Picasso: A tool for live package dependendy analysis}
\label{picasso_overview}
\end{figure}

\paragraph{}
Picasso applies the proposed refactoring algorithm on software packages and provides a list of class moving suggestions. An example of the suggestions that Picasso presents is depicted in Fig. \ref{picasso_suggestions}. Every suggestion is a class movement from a source package to a target package. 

\begin{figure}[ht!]
\centering
\includegraphics[scale=0.5]{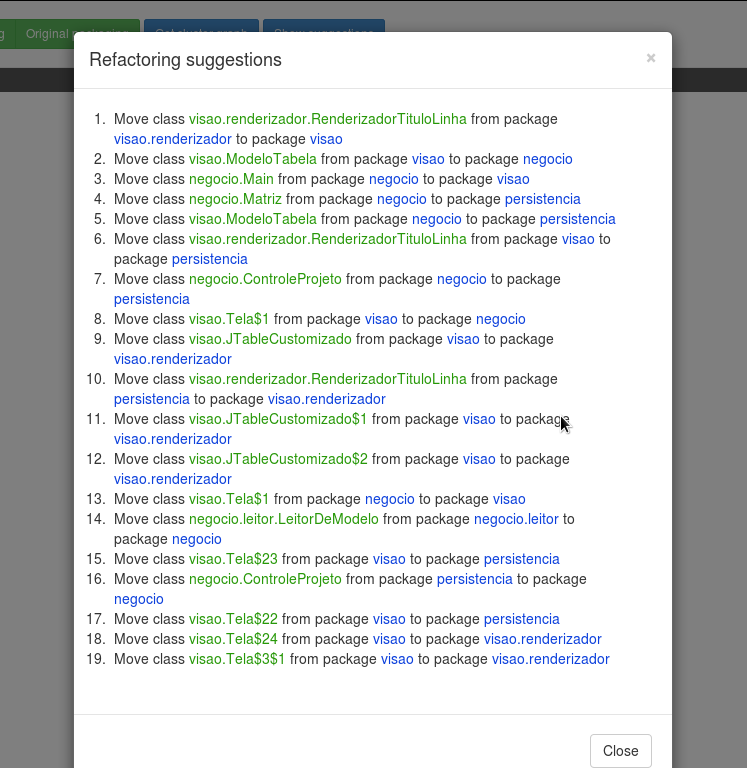}
\caption{An example of some suggestions provided by Picasso}
\label{picasso_suggestions}
\end{figure}

\paragraph{} 
Picasso provides many extra features that are as follows.

\begin{itemize}
\item Import Java jar files and class files.
\item Import UML structures.
\item Provides an option to choose famous graphs such as the Zachary club network.
\item Calculates modularity and provides a refactored solution for a software system using Alg. \ref{alg:proposed_refactoring}.
\item Calculates Martin's instability metric for software packages.
\item Hierarchically provides cluster graphs of a graph.
\item Provides an extendible messaging system for future works.
\item Provides an edited version of JSNetworkX's force layout graph visualization algorithm.
\item Provides functions for adding and removing graph edges and nodes.
\item Provides the ability to lock graph nodes in one position for better viewing.
\end{itemize}

\paragraph{}
Picasso's top menu provides the main functionalities of the tool. The menu bar is depicted in Fig. \ref{picasso_menu} and shows that the tool is in working mode and awaits a response from the Picasso server. The gray section of the top bar shows some information such as the modularity measure of the current clustering and the name of the current selected class in the dependency graph. The top buttons consist of two main groups. The left, green buttons provide directed refactoring, undirected refactoring and the original clustering of the software system being analyzed. The right, blue buttons provide the options for viewing the graph's clustering graph, viewing the movement suggestions after refactoring and viewing instability measures for different packages. An example of the instability measure window is shown in Fig. \ref{picasso_inst}.

\begin{figure}[ht!]
\centering
\includegraphics[scale=0.55]{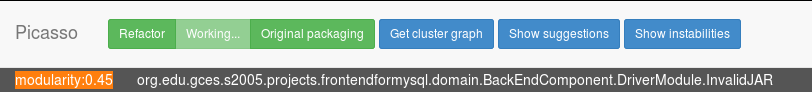}
\caption{Picasso's top menu bar}
\label{picasso_menu}
\end{figure}

\begin{figure}[ht!]
\centering
\includegraphics[scale=0.5]{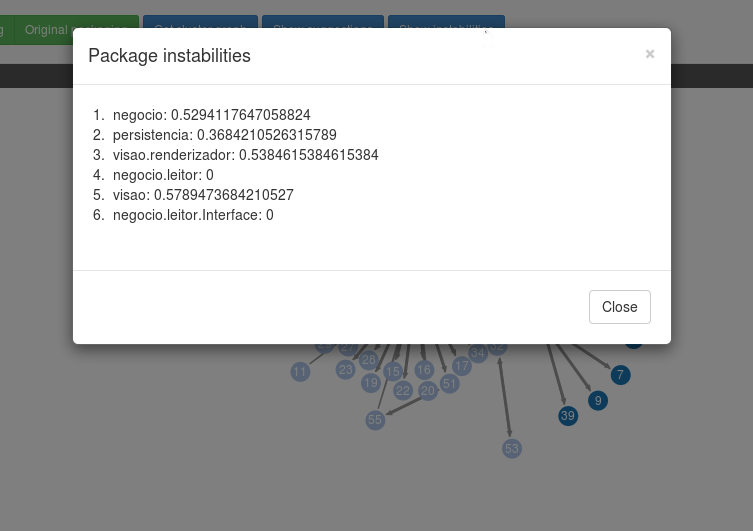}
\caption{An example of Picasso's instabilities window}
\label{picasso_inst}
\end{figure}

\begin{figure}[ht!]
\centering
\includegraphics[scale=0.4]{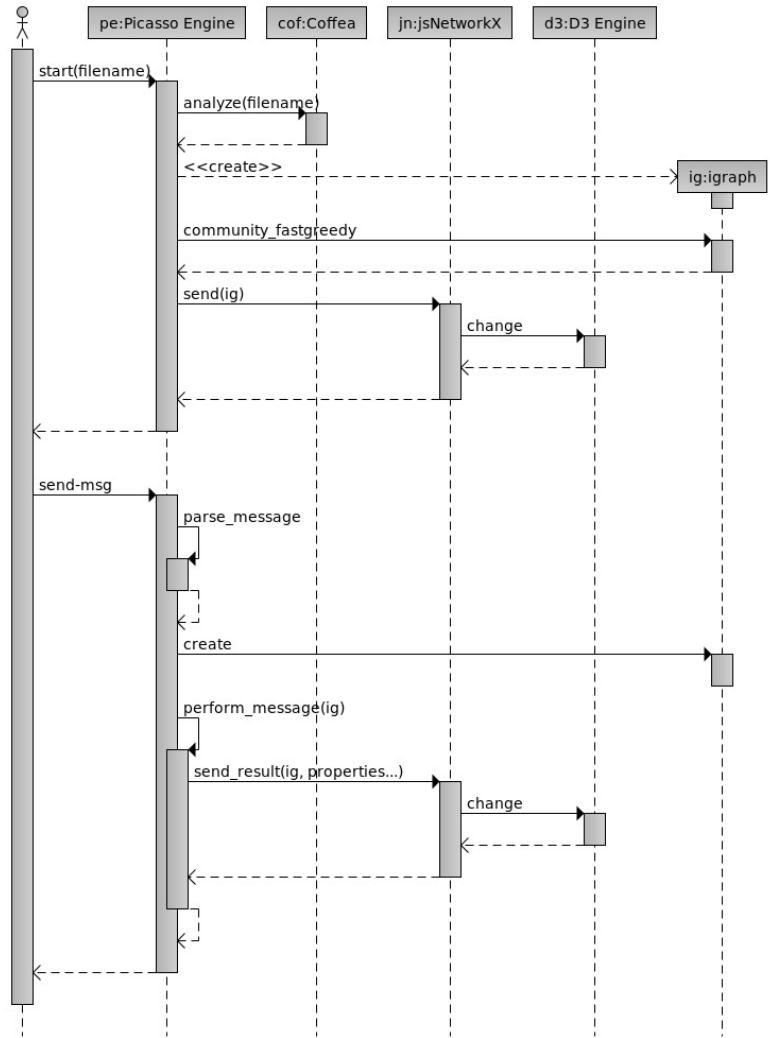}
\caption{Picasso's sequence diagram}
\label{seq_diagram}
\end{figure}

\section{Picasso's 3rd party dependencies}
\paragraph{}
Picasso utilizes many diverse 3rd party libraries. Some of these libraries have been customized and tweaked specially for Picasso. The following list contains some brief information on these libraries.

\begin{itemize}
\item \textbf{Coffea\footnote{\url{https://github.com/sbilinski/coffea}} java analysis tool.} Coffea is an open source static code analyzer for Java byte code that can export package dependency graphs in various graph file formats. Coffea is written in Python and therefore can be integrated well with Picasso.
\item \textbf{D3\footnote{\url{http://d3js.org}} visualization library.} D3 stands for Data-Driven Documents, and is arguably one of the best Javascript data visualization tools that utilizes HTML5, SVG (Scalable Vector Graphics), CSS3 and Javascript capabilities and provides an extremely flexible platform for data visualization. 
\item \textbf{JSNetworkX\footnote{\url{http://felix-kling.de/JSNetworkX}} network visualization library.} This library is a port of the popular NetworkX Python graph library and is build upon the D3 platform. 
\item \textbf{Python's igraph\footnote{\url{http://igraph.org}} library.} Python's igraph library is used in Picasso for creating and manipulating graphs on the server side.
\item \textbf{Python's Socks-js\footnote{\url{https://github.com/sockjs/sockjs-client}} library.} The Socks-JS library is used by Picasso for creating a web socket messaging system that can pass graph and graph cluster information between the server and client sides of the program.
\end{itemize}

\paragraph{}
The sequence diagram in Fig. \ref{seq_diagram} shows how Picasso interacts with these dependencies. 
\chapter{Conclusion and future works}
\pagebreak

\paragraph{}
The fast expansion of software systems and their complexities, makes large software projects difficult to maintain, and their components hard to reuse. The focus of this work is to use the benefits of graph clustering algorithms and present a refactoring technique for software packages while considering several important software metrics such as coupling, cohesion and stability. 

\paragraph{}
This work presents a proposition and proof that the cluster quality metric provided by Newman \cite{newman2004finding} is in favor of Martin's Stable Dependencies Principle \cite{martin2003agile} and provides examples that show how directed graphs are important when a system is being modeled with dependency graphs. 

\paragraph{}
For evaluating our proposed algorithm and to test it in a real life scenario, we implemented a tool, namely Picasso for refactoring software packages and visualizing their directed dependency graphs. The provides tool takes a software system written in the Java language and gives a list of suggested movements for classes. 

\paragraph{}
Some ideas are presented in the following sections as future works. These possible works are divided into two main categories; refactoring and tool improvements. 

\section{Refactoring}
\paragraph{}
The refactoring method presented in this work utilizes a directed and weighted version of Newman's modularity. This requires modularity to be calculated in every step of the proposed algorithm and thus performs slower than the algorithm of Pan et al \cite{pan2013refactoring}. This may be considered as one of the problems that can be tackled in future works. Also, some rare problems have been found with the directed version of modularity \cite{malliaros2013clustering} and alternative approaches should also be considered, i.e. random walk based mathods such as LinkRank.

\paragraph{}
The importance of directed dependency graphs can also be analyzed in the class level, while using an appropriate metric for class couplings and cohesion. 

\section{Tool improvements}
\paragraph{}
Some improvements can be applied on the tool proposed in this work. Currently a force directed layout is used for visualizing graphs. A force directed layouts simulate physical forces between nodes and edges to aesthetically draw a graph. Spring like attractive forces that are based on Hooke's law are typically used. The force directed layout can be enhanced with collision detection algorithms, so that nodes that are members of the same community can be grouped together instead of being mixed in with nodes from other communities. Also several problems with force directed layouts in large graphs have been pointed out in the literature \cite{yakovlev2009cluster} and radial tree layouts have been proposed as alternatives. Radial tree layouts can be considered in future implementations of the tool. An example of a radial tree layout from a tool named Barrio, provided in \cite{yakovlev2009cluster} is depicted in Fig. \ref{radial_tree}.

\begin{figure}[ht!]
\centering
\includegraphics[scale=0.5]{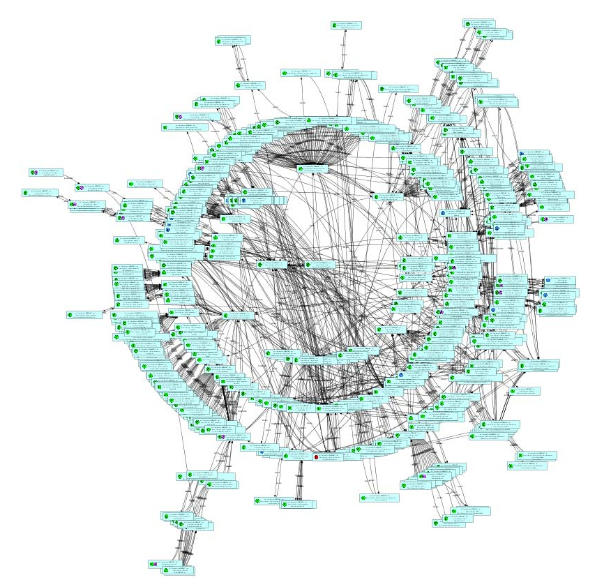}
\caption{An example of a radial tree layout}
\label{radial_tree}
\end{figure}

\paragraph{}
Being able to force a node to be a member of a certain community while calculating the resulting modularity of the graph cluster can be considered as one of the important options in future versions of the application. Some library classes might need to be kept in their original package even though modularity is decreased by doing so. 

\begin{appendices}
\chapter{Server side code for Picasso}
\begin{lstlisting}[language=python]
# -*- coding: utf-8 -*-
"""
Picasso, ver 0.1
Author: Mohammad A.Raji
Depends on:
    -sockjs-tornado for the asynchronus python server
    -D3.js for visualizing graphs
    -JSNetworkX for visualizing graphs
    -igraph for community detection algorithms
    -Coffea for extracting java dependencies
"""
from __future__ import division
import os
import tornado.ioloop
import tornado.web
import sockjs.tornado
import igraph
import time
import json
import hashlib
from igraph import *

# Request handles class for the index page
class IndexHandler(tornado.web.RequestHandler):
    def get(self):
        self.render('picasso.html')

# Connection class: responsible for all the client/server connections
class Connection(sockjs.tornado.SockJSConnection):
    participants = set()

    def on_open(self, info):
        # Add client to the clients list
        self.participants.add(self)
        if len(sys.argv) > 1:
            if len(sys.argv) > 2:
                if sys.argv[1] == "--famous":
                    g = Graph.Famous(sys.argv[2])
                    #g.to_undirected()
                    self.broadcast(self.participants, "graph/" + refactoring.graphToString())
            else:
                refactoring.parseCode(sys.argv[1]);
                self.broadcast(self.participants, "graph/" + refactoring.graphToString())
                self.broadcast(self.participants, "labels/" + refactoring.getVertexLabels())

    def on_message(self, message):
        # Take appropriate action when a message arrives from the client
        self.parseAndApplyMessage(message)

    def on_close(self):
        # Remove client from the clients list and broadcast leave message
        self.participants.remove(self)

    def parseAndApplyMessage(self, msg):
        global refactoring
        message = msg.split("/")
        command = message[0]
        if (len(message) > 1):
            argument = message[1];

            if command == "clusters":
                refactoring.parseGraph(argument)
                refactoring = Refactoring(refactoring.detectCommunities().cluster_graph())
                self.broadcast(self.participants, "graph/" + refactoring.graphToString())
            elif command in ["addnode", "removenode", "addedge", "removeedge"]:
                refactoring.parseChange(command, argument)
            elif command == "getoriginal":
                self.broadcast(self.participants, "membership/" + str(refactoring.original_membership).strip("[]"))
                self.broadcast(self.participants, "measures/" + Refactoring.formatMeasures(refactoring.original_modularity));
                self.broadcast(self.participants, "instability/" + Refactoring.formatInstability(refactoring.original_package_instability));
            elif command == "refactor":
                refactored_results = refactoring.refactor()
                go_membership = refactored_results[1]
                self.broadcast(self.participants, "membership/" + str(go_membership).strip("[]"))
                self.broadcast(self.participants, "measures/" + Refactoring.formatMeasures(refactored_results[0]));
                self.broadcast(self.participants, "suggestions/" + Refactoring.formatSuggestions(refactored_results[2]));
                self.broadcast(self.participants, "instability/" + Refactoring.formatInstability(refactoring.getInstabilityForEachPackage(refactoring.g, go_membership, refactoring.packages)));
            elif command == "urefactor":
                    refactored_results = refactoring.refactor(False)
                    go_membership = refactored_results[1]
                    self.broadcast(self.participants, "membership/" + str(go_membership).strip("[]"))
                    self.broadcast(self.participants, "measures/" + Refactoring.formatMeasures(refactored_results[0]));
                    self.broadcast(self.participants, "suggestions/" + Refactoring.formatSuggestions(refactored_results[2]));
                    self.broadcast(self.participants, "instability/" + Refactoring.formatInstability(refactoring.getInstabilityForEachPackage(refactoring.g, go_membership, refactoring.packages)));

            elif command == "fastgreedy":
                go_membership = refactoring.detectCommunities().membership
                self.broadcast(self.participants, "membership/" + str(go_membership).strip("[]"))
                self.broadcast(self.participants, "measures/" + refactoring.getClusterMeasures());

        else:
            refactoring.parseGraph(command)
            go_membership = refactoring.detectCommunities().membership
            self.broadcast(self.participants, "membership/" + str(go_membership).strip("[]"))
            self.broadcast(self.participants, "measures/" + refactoring.getClusterMeasures());

# All refactoring and graph related capabilities
class Refactoring():
    g = None;
    gc = None;
    parsed_code_filename = None;
    original_membership = None;
    original_modularity = None;

    def __init__(self, graph=None):
        self.packages = dict()
        self.original_package_instability = dict()
        self.g = graph

    def parseChange(self, command, arg):
        if command == "addnode":
            self.g.add_vertex(arg)
        elif command == "removenode":
            self.g.delete_vertices(arg)
        elif command == "addedge":
            from_edge = int(arg.split(",")[0])
            to_edge = int(arg.split(",")[1])
            self.g.add_edge(from_edge, to_edge)
        elif command == "removeedge":
            from_edge = int(arg.split(",")[0])
            to_edge = int(arg.split(",")[1])
            self.g.delete_edges((from_edge, to_edge))

    def parseGraph(self, st):
        self.g = Graph()
        st_graph = st.split("|")
        vertices = st_graph[0].split(";")
        for v in vertices:
            self.g.add_vertex(v)

        edges = st_graph[1].split(";")
        for e in edges:
            from_edge = e.split(",")[0]
            to_edge = e.split(",")[1]
            self.g.add_edge(from_edge, to_edge)
        return self.g

    def graphToString(self):
        vertexlist = []
        for v in self.g.vs:
            vertexlist.append(v.index)
        vertex_str = str(vertexlist).strip("[]");
        vertex_str = vertex_str.replace(" ", "");
        s = str(self.g.get_edgelist()).strip("[]");
        s = s.replace("(", "");
        s = s.replace("),", ";");
        s = s.replace(")", "");
        s = s.replace(" ", "");
        s = vertex_str + "|" + s
        return s

    @staticmethod
    def formatMeasures(measure):
        measures = "";
        measures += "modularity:" + str(round(measure, 2)) # + ","
        return measures;

    @staticmethod
    def formatSuggestions(suggestions):
        return json.dumps(suggestions)

    @staticmethod
    def formatInstability(package_instability):
        return json.dumps(package_instability.items())

    def getClusterMeasures(self):
        if (self.gc == None):
            self.gc = self.detectCommunities()
        measures = "";
        measures += "modularity:" + str(round(self.gc.modularity, 2)) # + ","
        return measures;

    def getVertexLabels(self):
        msg = "";
        for v in self.g.vs:
            msg = msg + v['label'] + ",";

        msg = msg.strip(",")
        return msg

    def detectCommunities(self):
        self.g = self.g.simplify(loops='False', multiple='False')
        gc = self.g.as_undirected().community_fastgreedy()
        gc = gc.as_clustering()
        self.gc = gc
        return gc

    # This function works independently from local graph g
    def makeDwpdnMembership(self, graph):
        # If this graph has no label attribute at all
        if "label" not in graph.vs.attribute_names():
            custom_package_index = 0
            for v in graph.vs:
                v['label'] = str(custom_package_index) + "."
                custom_package_index += 1

        self.packages = dict()
        membership = []
        recent_package = 0;
        for v in graph.vs:
            if v['label'] == None:
                # Make a random package name if this package is a new isolated
                # node with no name
                random_package_name = hashlib.md5(str(time.time())).hexdigest()[0:5] + "."
                v['label'] = random_package_name
            package_name = v['label'].rsplit(".", 1)[0];
            if (self.packages.has_key(package_name)):
                membership.append(self.packages[package_name]);
            else:
                self.packages[package_name] = recent_package;
                membership.append(recent_package);
                recent_package += 1;

        return membership

    # This function works independently from local graph g
    def calculateQ(self, graph, membership):
        Q = 0.0;
        graph = graph.simplify(loops='False', multiple='False')
        m = graph.ecount();
        edge_count_factor = 2*m;
        if graph.is_directed() == True:
            edge_count_factor = m
        for i in graph.vs:
            for j in graph.vs:
                if membership[i.index] != membership[j.index]:
                    continue
                else:
                    Aij = 0
                    if graph.are_connected(i, j):
                        Aij = 1
                    wi_out = i.outdegree()
                    wj_in = j.indegree()
                    Q += Aij - (wi_out*wj_in) / edge_count_factor

        Q *= 1/edge_count_factor
        return Q;

    def getInstabilityForEachPackage(self, graph, membership, packages):
        package_in = packages.fromkeys(packages.iterkeys(), 0);
        package_out = packages.fromkeys(packages.iterkeys(), 0);
        package_instability = packages.fromkeys(packages.iterkeys(), 0);
        for e in graph.es:
            if (membership[e.target] != membership[e.source]):
                package_in[self.getPackageNameFromIndex(membership[e.target])] += 1;
                package_out[self.getPackageNameFromIndex(membership[e.source])] += 1;

        print package_out
        print package_in
        for package in packages:
            if package_out[package] + package_in[package] != 0:
                package_instability[package] = package_out[package] / (package_out[package] + package_in[package])
            else:
                package_instability[package] = 0;

        return package_instability

    def getPackageNameFromIndex(self, index):
        for name, i in self.packages.iteritems():
            if i == index:
                return name

    def refactor(self, directed=True):
        if directed == True:
            graph = self.g;
        else:
            graph = self.g.as_undirected()
        suggested_movements = [];
        Q_prime = -1;
        membership = self.makeDwpdnMembership(graph);
        # Check if there is only one package
        if membership.count(0) == len(membership):
            membership = range(len(membership))
        Q = self.calculateQ(graph, membership);
        selected_community = 0;
        v_range = range(graph.vcount())
        while True:
            restart_loop = False
            for index in v_range:
                i = graph.vs[index]
                for j in graph.vs:
                    temp_membership = list(membership)
                    temp_membership[i.index] = temp_membership[j.index];
                    temp_Q = self.calculateQ(graph, temp_membership);
                    if (temp_Q > Q_prime):
                        Q_prime = temp_Q;
                        selected_community = membership[j.index];
                if (Q_prime > Q):
                    suggested_movements.append((i['label'], self.getPackageNameFromIndex(membership[i.index]), self.getPackageNameFromIndex(selected_community)));
                    membership[i.index] = selected_community;
                    Q = Q_prime;
                    restart_loop = True
                    break;
            if not restart_loop:
                break;

        print "Done refactoring"
        return [Q_prime, membership, suggested_movements]

    def parseCode(self, filename):
        os.system("coffea -R -i " + filename + " -f gml -o temp.gml")
        self.parsed_code_filename = filename
        self.g = read('temp.gml')
        self.original_membership = self.makeDwpdnMembership(self.g)
        self.original_modularity = self.calculateQ(self.g, self.original_membership);
        self.original_package_instability = self.getInstabilityForEachPackage(self.g, self.original_membership, self.packages)

if __name__ == "__main__":
    import logging
    logging.getLogger().setLevel(logging.DEBUG)

    # Instantiate the main refactoring object
    refactoring = Refactoring();

    # Create the router
    Router = sockjs.tornado.SockJSRouter(Connection, '/picasso')

    # Create Tornado application
    app = tornado.web.Application(
            [(r"/", IndexHandler)] + Router.urls
    )

    # Make Tornado app and listen on port 8081
    port = 8081
    app.listen(port)
    print "Listening on port " + str(port);

    # Start IOLoop
    tornado.ioloop.IOLoop.instance().start()

\end{lstlisting} 

\chapter{Client side Javascript of Picasso}
\begin{lstlisting}[language=javascript]
    last = 1;
    conn = null;
    labels = [];
    $(function() {
      colors = ['#FF7F0E', '#AEC7E8', '#2CA02C', '#D62728', '#1F77B4']
      color = window.d3.scale.category20();
      function log(msg)
      {
          console.log(msg);
      }

      function parseAndApplyMessage(msg)
      {
          var message = msg.split("/");
          var command = message[0];
          if (message.length > 1)
          {
              var arguments = message[1];
              if (command == "membership")
              {
                  applyMembership(arguments);
              }
              else if (command == "graph")
              {
                  applyGraph(arguments);
              }
              else if (command == "labels")
              {
                  saveLabels(arguments);
              }
              else if (command == "measures")
              {
                  updateMeasures(arguments);
              }
              else if (command == "suggestions")
              {
                  setSuggestions(arguments);
                  $("#refactor-btn").button("reset");
                  $("#urefactor-btn").button("reset");
              }
              else if (command == "instability")
              {
                  setInstabilities(arguments);
              }
          }
      }

      function setSuggestions(msg)
      {
          var suggestions = JSON.parse(msg);
          txt = "<ol>";
          for (i in suggestions)
          {
              txt += "<li>Move class <span class='class-name'>" + suggestions[i][0] + "</span> from package <span class='package-name'>" + suggestions[i][1] + "</span> to package <span class='package-name'>" + suggestions[i][2] + "</span></li>";
          }
          txt += "</ol>"
          $(".modal-body").html(txt);
      }

      function setInstabilities(msg)
      {
          var instabilities = JSON.parse(msg);
          txt = "<ol>";
          for (i in instabilities)
          {
              txt += "<li>" + instabilities[i][0] + ": " + instabilities[i][1] + "</li>";
          }
          txt += "</ol>"
          $("#instability .modal-body").html(txt);
    }
      function updateMeasures(msg)
      {
          $("#measures").text(msg);
      }

      function saveLabels(msg)
      {
          labels = msg.split(",");
          for (var i = 0;i<G.nodes().length;i++)
          {
              d3.select("#node" + i.toString()).attr("data-label", labels[i]);
          }
      }
      function applyGraph(msg)
      {
          G.clear();
          var splitted_str = msg.split("|")
          var vertex_str = splitted_str[0];
          var edges_str = splitted_str[1];
          var vertices = vertex_str.split(",");
          for (key in vertices)
          {
              var vertex = parseInt(vertices[key]);
              G.add_node(vertex);
          }

          var edges = edges_str.split(";");
          for (key in edges)
          {
              var edge = edges[key];
              var from = parseInt(edge.split(",")[0]);
              var to = parseInt(edge.split(",")[1]);
              G.add_edge(from, to);
          }
          window.d3.selectAll(".node").on("mouseover", function(){
              jQuery("#label").text(d3.select(this).attr("data-label"));
          });
      }

      function applyMembership(msg)
      {
          membership = msg.split(", ");
          iteration = G.nodes_iter();
          for (key in membership)
          {
              node = iteration.next();
              G.node.get(node).color = color(membership[key]);
              window.d3.select("#node" + node.toString() + " circle" ).style("fill", color(membership[key]));
          }
      }

      function connect() {
        disconnect();

        var transports = $('#protocols input:checked').map(function(){
            return $(this).attr('id');
        }).get();

        conn = new SockJS('http://' + window.location.host + '/picasso', transports);

        log('Connecting...');

        conn.onopen = function() {
          log('Connected.');
          update_ui();
        };

        conn.onmessage = function(e) {
            //log('Received: ' + e.data);
            parseAndApplyMessage(e.data);
        };

        conn.onclose = function() {
          log('Disconnected.');
          conn = null;
          update_ui();
        };
      }

      function disconnect() {
        if (conn != null) {
          log('Disconnecting...');

          conn.close();
          conn = null;

          update_ui();
        }
      }

      function update_ui() {
        var msg = '';

        if (conn == null || conn.readyState != SockJS.OPEN) {
          $('#status').text('disconnected');
          $('#connect').text('Connect');
        } else {
          $('#status').text('connected (' + conn.protocol + ')');
          $('#connect').text('Disconnect');
        }
      }

      $('#connect').click(function() {
        if (conn == null) {
          connect();
        } else {
          disconnect();
        }

        update_ui();
        return false;
      });

      connect();
      
	addNode = function(node_name)
	{
		G.add_node(node_name);
		conn.send(G.nodes().join(";") + "|" + G.edges().join(";"));
	}
    addEdge = function(from, to)
    {
		G.add_edge(from, to);
		conn.send(G.nodes().join(";") + "|" + G.edges().join(";"));
    }
    removeEdge = function(from, to)
    {
		G.remove_edge(from, to);
		conn.send(G.nodes().join(";") + "|" + G.edges().join(";"));
    }
    drawGraph = function()
    {
        jsnx.draw(G, {
            element: '#canvas',
            with_labels: true,
            pan_zoom: {
                enabled: false
            },
            layout_attr: {
            	'charge': -420,
                'linkDistance': 100
            },
            node_style: {
                fill: function(d) 
                {
                    return d.data.color || '#AAA';
                },
                stroke: 'none'
            },
            edge_style: {
                fill: '#999'
            },
            label_style: {
                fill: 'white',
                'font-size': '12px'
            }
        }, true);
    }
	G = jsnx.DiGraph();
    drawGraph();
    });
\end{lstlisting}

\end{appendices}

\bibliographystyle{ieeetr}
\bibliography{thesis1.bib}

\begin{thebibliography}{10}

\bibitem{sommerville2004}
I.~Sommerville, {\em Software Engineering. International computer science
  series}.
\newblock Addison Wesley, 2004.

\bibitem{lim1994}
W.~C. Lim, ``Effects of reuse on quality, productivity, and economics,'' {\em
  Software, IEEE}, vol.~11, no.~5, pp.~23--30, 1994.

\bibitem{fowler1999}
M.~Fowler, {\em Refactoring: improving the design of existing code}.
\newblock Addison-Wesley Professional, 1999.

\bibitem{brodie1984}
L.~Brodie, {\em Thinking Forth}.
\newblock Prentice Hall, 1984.

\bibitem{goto2013extract}
A.~Goto, N.~Yoshida, M.~Ioka, E.~Choi, and K.~Inoue, ``How to extract
  differences from similar programs? a cohesion metric approach,'' in {\em
  Software Clones (IWSC), 2013 7th International Workshop on}, pp.~23--29,
  IEEE, 2013.

\bibitem{pan2009class}
W.~Pan, B.~Li, Y.~Ma, J.~Liu, and Y.~Qin, ``Class structure refactoring of
  object-oriented softwares using community detection in dependency networks,''
  {\em Frontiers of Computer Science in China}, vol.~3, no.~3, pp.~396--404,
  2009.

\bibitem{baggen2012standardized}
R.~Baggen, J.~P. Correia, K.~Schill, and J.~Visser, ``Standardized code quality
  benchmarking for improving software maintainability,'' {\em Software Quality
  Journal}, vol.~20, no.~2, pp.~287--307, 2012.

\bibitem{kitchenham2010s}
B.~Kitchenham, ``What’s up with software metrics?--a preliminary mapping
  study,'' {\em Journal of Systems and Software}, vol.~83, no.~1, pp.~37--51,
  2010.

\bibitem{fenton1998software}
N.~E. Fenton and S.~L. Pfleeger, {\em Software metrics: a rigorous and
  practical approach}.
\newblock PWS Publishing Co., 1998.

\bibitem{troy1981measuring}
D.~A. Troy and S.~H. Zweben, ``Measuring the quality of structured designs,''
  {\em Journal of Systems and Software}, vol.~2, no.~2, pp.~113--120, 1981.

\bibitem{alghamdi2008measuring}
J.~S. Alghamdi, ``Measuring software coupling.,'' {\em Arabian Journal for
  Science \& Engineering (Springer Science \& Business Media BV)}, vol.~33,
  2008.

\bibitem{chidamber1991towards}
S.~R. Chidamber and C.~F. Kemerer, {\em Towards a metrics suite for object
  oriented design}, vol.~26.
\newblock ACM, 1991.

\bibitem{hitz1995measuring}
M.~Hitz and B.~Montazeri, ``Measuring coupling and cohesion in object-oriented
  systems,'' in {\em Proceedings of the International Symposium on Applied
  Corporate Computing}, vol.~50, pp.~75--76, 1995.

\bibitem{representational2003}
P.~Kelsen, ``An information-based view of representational coupling in
  object-oriented systems,'' in {\em Fundamental Approaches to Software
  Engineering} (M.~Pezzè, ed.), vol.~2621 of {\em Lecture Notes in Computer
  Science}, pp.~216--230, Springer Berlin Heidelberg, 2003.

\bibitem{martin2003agile}
R.~C. Martin, {\em Agile software development: principles, patterns, and
  practices}.
\newblock Prentice Hall PTR, 2003.

\bibitem{wiggerts1997using}
T.~A. Wiggerts, ``Using clustering algorithms in legacy systems
  remodularization,'' in {\em Reverse Engineering, 1997. Proceedings of the
  Fourth Working Conference on}, pp.~33--43, IEEE, 1997.

\bibitem{chen2013new}
M.~Chen, T.~Nguyen, and B.~K. Szymanski, ``A new metric for quality of network
  community structure,'' {\em HUMAN}, vol.~2, no.~4, pp.~pp--226, 2013.

\bibitem{shtern2012clustering}
M.~Shtern and V.~Tzerpos, ``Clustering methodologies for software
  engineering,'' {\em Advances in Software Engineering}, vol.~2012, p.~1, 2012.

\bibitem{pan2013refactoring}
W.-F. Pan, B.~Jiang, and B.~Li, ``Refactoring software packages via community
  detection in complex software networks,'' {\em International Journal of
  Automation and Computing}, vol.~10, no.~2, pp.~157--166, 2013.

\bibitem{newman2004finding}
M.~E. Newman and M.~Girvan, ``Finding and evaluating community structure in
  networks,'' {\em Physical review E}, vol.~69, no.~2, p.~026113, 2004.

\bibitem{clauset2004finding}
A.~Clauset, M.~E. Newman, and C.~Moore, ``Finding community structure in very
  large networks,'' {\em Physical review E}, vol.~70, no.~6, p.~066111, 2004.

\bibitem{girvan2002community}
M.~Girvan and M.~E. Newman, ``Community structure in social and biological
  networks,'' {\em Proceedings of the National Academy of Sciences}, vol.~99,
  no.~12, pp.~7821--7826, 2002.

\bibitem{pons2005computing}
P.~Pons and M.~Latapy, ``Computing communities in large networks using random
  walks,'' in {\em Computer and Information Sciences-ISCIS 2005}, pp.~284--293,
  Springer, 2005.

\bibitem{malliaros2013clustering}
F.~D. Malliaros and M.~Vazirgiannis, ``Clustering and community detection in
  directed networks: A survey,'' {\em Physics Reports}, vol.~533, no.~4,
  pp.~95--142, 2013.

\bibitem{arenas2007size}
A.~Arenas, J.~Duch, A.~Fern{\'a}ndez, and S.~G{\'o}mez, ``Size reduction of
  complex networks preserving modularity,'' {\em New Journal of Physics},
  vol.~9, no.~6, p.~176, 2007.

\bibitem{murphy1995software}
G.~C. Murphy, D.~Notkin, and K.~Sullivan, ``Software reflexion models: Bridging
  the gap between source and high-level models,'' in {\em ACM SIGSOFT Software
  Engineering Notes}, vol.~20, pp.~18--28, ACM, 1995.

\bibitem{mancoridis1999bunch}
S.~Mancoridis, B.~S. Mitchell, Y.~Chen, and E.~R. Gansner, ``Bunch: A
  clustering tool for the recovery and maintenance of software system
  structures,'' in {\em Software Maintenance, 1999.(ICSM'99) Proceedings. IEEE
  International Conference on}, pp.~50--59, IEEE, 1999.

\bibitem{sadayappan1990cluster}
P.~Sadayappan, F.~Ercal, and J.~Ramanujam, ``Cluster partitioning approaches to
  mapping parallel programs onto a hypercube,'' {\em Parallel Computing},
  vol.~13, no.~1, pp.~1--16, 1990.

\bibitem{tang2012survey}
Y.~Tang, K.~Sycara, M.~Sensoy, and J.~Z. Pan, ``Survey on clustering methods
  for ontological knowledge,'' 2012.

\bibitem{ghafourian2013modularization}
S.~Ghafourian, A.~Rezaeian, and M.~Naghibzadeh, ``Modularization of
  graph-structured ontology with semantic similarity,'' in {\em Workshop on
  Modular Ontologies (WoMO) 2013}, p.~25, 2013.

\bibitem{melton2006identifying}
H.~Melton and E.~Tempero, ``Identifying refactoring opportunities by
  identifying dependency cycles,'' in {\em Proceedings of the 29th Australasian
  Computer Science Conference-Volume 48}, pp.~35--41, Australian Computer
  Society, Inc., 2006.

\bibitem{gupta2009package}
V.~Gupta and J.~K. Chhabra, ``Package coupling measurement in object-oriented
  software,'' {\em Journal of computer science and technology}, vol.~24, no.~2,
  pp.~273--283, 2009.

\bibitem{briand1996property}
L.~C. Briand, S.~Morasca, and V.~R. Basili, ``Property-based software
  engineering measurement,'' {\em Software Engineering, IEEE Transactions on},
  vol.~22, no.~1, pp.~68--86, 1996.

\bibitem{newman2004fast}
M.~E. Newman, ``Fast algorithm for detecting community structure in networks,''
  {\em Physical review E}, vol.~69, no.~6, p.~066133, 2004.

\bibitem{fortunato2010community}
S.~Fortunato, ``Community detection in graphs,'' {\em Physics Reports},
  vol.~486, no.~3, pp.~75--174, 2010.

\bibitem{yakovlev2009cluster}
V.~Yakovlev, ``Cluster analysis of object-oriented programs: a thesis presented
  in partial fulfilment of the requirements for the degree of master of science
  in computer science at massey university, palmerston north, new zealand,''
  2009.

\end{thebibliography}

%
%
%


\end{document}